\documentclass[11pt,a4paper]{article}
\usepackage{amssymb,amsmath}
\usepackage{color}
\usepackage{graphicx,graphics}
\usepackage{mathtools}
\usepackage[english]{babel}
\usepackage[utf8]{inputenc}
\usepackage{epsfig,url}
\usepackage{bbm,theorem}
\usepackage{a4wide}
\usepackage{enumerate}

\newtheorem{theorem}{Theorem}[section]
\newtheorem{definition}[theorem]{Definition}
\newtheorem{corollary}[theorem]{Corollary}
\newtheorem{lemma}[theorem]{Lemma}
\newtheorem{proposition}[theorem]{Proposition}
\newtheorem{assumption}[theorem]{Assumption}

{\theorembodyfont{\upshape}

}
\numberwithin{equation}{section}
\numberwithin{theorem}{section}

\newcommand{\qed}{\hfill$\Box$}
\newenvironment{proof}{\begin{trivlist}\item[]{\em Proof:}\/}{%
\qed\end{trivlist}}
\newenvironment{proofof}[1]{%
\begin{trivlist}\item[]{\em Proof of #1}\ }{\qed\end{trivlist}}

\newcommand{\E}{{\mathbb E}}
\newcommand{\Z}{{\mathbb Z}}

\newcommand{\R}{{\mathbb R}}
\newcommand{\C}{{\mathbb C\hspace{0.05 ex}}}
\newcommand{\N}{{\mathbb N}}
\newcommand{\T}{{\mathbb T}}

\newcommand{\cf}[1]{{\mathbbm 1}_{\{#1\}}}

\newcommand{\ci}{{\rm i}}
\newcommand{\re}{{\rm Re\,}}
\newcommand{\im}{{\rm Im\,}}
\newcommand{\rme}{{\rm e}}
\newcommand{\rmd}{{\rm d}}

\newcommand{\FT}[1]{\widehat{#1}}
\newcommand{\IFT}[1]{\widetilde{#1}}

\newcommand{\norm}[1]{\Vert #1\Vert}
\newcommand{\defset}[2]{ \left\{ #1\left|\,
 #2\makebox[0cm]{$\displaystyle\phantom{#1}$}\right.\!\right\} }
\newcommand{\set}[1]{\{#1\}}

\newcommand{\mean}[1]{\langle #1\rangle}
\newcommand{\bigmean}[1]{\left\langle #1\right\rangle}
\newcommand{\vep}{\varepsilon}
\newcommand{\defem}[1]{{\em #1\/}}
\newcommand{\qand}{\quad\text{and}\quad}

\newcommand{\ommin}{\omega_{\text{min}}}

\newcommand{\rhoc}{\rho_{\mathrm{c}}}

\newcounter{jlisti}

\title{Multi-state condensation in Berlin--Kac spherical models}
\author{Jani Lukkarinen\thanks{\emailjani}\\[0.5em]
$\,^*$\UHaddress}
\date{\today}
\newcommand{\email}[1]{E-mail: \tt #1}
\newcommand{\emailjani}{\email{jani.lukkarinen@helsinki.fi}}
\newcommand{\UHaddress}{\em University of Helsinki, Department of Mathematics and Statistics\\
\em P.O. Box 68, FI-00014 Helsingin yliopisto, Finland}

\begin{document}

\maketitle

\begin{abstract}
 We consider the Berlin--Kac spherical model for supercritical densities under a periodic lattice energy function which has finitely many non-degenerate global minima.
 Energy functions arising from nearest neighbour interactions on a rectangular lattice have a unique minimum, and in that case the supercritical fraction of the total mass condenses to the ground state of the energy function.
 We prove that for any sufficiently large lattice size this also happens in the case of multiple global minima, although the precise distribution 
 of the supercritical mass and the structure of the condensate mass fluctuations 
 may depend on the lattice size.   
 However, in all of these cases, one can identify a bounded number of degrees of freedom forming the condensate in such a way that their fluctuations 
 are independent from the rest of the fluid.
 More precisely, the original Berlin--Kac measure may  be replaced by a measure where the condensate and normal fluid degrees of freedom become independent random variables, and the normal fluid part converges to the critical Gaussian free field.
 The proof is based on a construction of a suitable coupling between the two measures, proving that their Wasserstein distance is small enough for
 the error in any finite moments of the field to vanish as the lattice size is increased to infinity.
\end{abstract}

\section{Introduction}

Berlin and Kac proposed \cite{BerlinKac52} in 1952 a {\em spherical model\/} as a modification of the Ising model 
of a ferromagnet.  In their model, discrete spin variables are replaced by continuum variables, i.e., by real numbers,
while keeping a constraint that the total length of the continuum vector equals that of the discrete spin vector.
This enforces the continuum spin vectors to remain on the surface of a fixed high-dimensional sphere, 
hence the name ``spherical model.''
Their motivation was to find simple models were phase transitions could be studied fairly explicitly, in particular, in the physically relevant case of three dimensions.  

Although the partition function of the spherical model cannot be explicitly solved for fixed finite lattices,
it has an integral representation which allows studying the properties of its infinite volume limit when restricted to nearest neighbour interactions.  
The limiting partition function is sufficiently explicit that standard thermal equilibrium properties of the model can be derived from it and, 
as shown in \cite{BerlinKac52}, the spherical model in three dimensions has a phase transition corresponding to spontaneous magnetisation.  
The reference also contains estimates for the second and fourth moments of the field, 
implying that the fluctuations at small temperatures, when there is spontaneous magnetisation, cannot be Gaussian.

On a technical level, the spontaneous magnetisation found in \cite{BerlinKac52} is analogous to Bose--Einstein
condensation in quantum statistical mechanics.  For instance, 
Yan and Wannier  \cite{YanWannier65} extend the analysis in \cite{BerlinKac52} to compute also the 
single site distribution (one-point function) in the infinite volume limit.  
They find that in the subcritical case the distribution is Gaussian whereas in the supercritical 
case it is not Gaussian but instead corresponds to a random variable which is 
a sum of a random constant and a Gaussian variable.  The appearance of the 
constant is analogous to the effect of condensation for ideal Bose gas.

To elucidate the connection further, let us begin with more detailed  definitions.
The spherical model in $d$ dimensions is defined as the random field of ``continuous spin''
$s_x\in \R$, $x\in \Lambda$, where $\Lambda\subset \R^d$ is a finite lattice of points.
The main purpose of using a lattice to label the spins is to define the interaction energy of a spin configuration:
one assumes that there is given a coupling function $J_{x,y}$, $x,y\in \Lambda$, such that the energy is given by 
\[
 E_\Lambda[s] := \sum_{x,y\in \Lambda} J_{x,y} s_x^* s_y\,,
\]
where $s_x^*$ denotes the complex conjugate, added here for later use.
Often one takes $J_{x,y} = v(x-y)$ for a function $v$ which decays sufficiently rapidly with increasing $|x-y|$.
For instance, the rectangular 
nearest neighbour case with Dirichlet boundary conditions would have $\Lambda \subset \Z^d$ and $v(x)=0$ for $|x|_\infty\ge 2$,
where $|x|_\infty := \max_i |x_i|$.  We will use both $|x|_\infty$ and the Euclidean norm on $\R^d$, $|x|$, frequently in the following.

Denoting the lattice size by $V=|\Lambda|<\infty$, the probability measure for the spin field $s$ at inverse temperature $\beta>0$ is given by
\begin{align}\label{eq:defBKmodel}
 \mu_{\text{BK},\beta}[\rmd s] = \frac{1}{Z_{\text{BK},\Lambda,\beta}} \rme^{-\beta E_\Lambda[s]}\,
 \delta\!\left(\sum_{x\in\Lambda}s_x^2 - V\right)\prod_{x\in \Lambda}\! \rmd s_x \,.
\end{align}
The first factor is the standard canonical Gibbs weight for the given temperature and energy function.  
The second ``factor'' is a $\delta$-function constraint which enforces the assumption that 
the length of the spin vector divided by the number of particles is 
equal to one. We will use such $\delta$-functions liberally in the following,
and the discussion about their mathematical definition and properties is given in Appendix \ref{sec:defdelta}.
In particular, it follows that under the above measure $\sum_{x\in\Lambda}s_x^2=V$ almost surely.
Here $Z_{\text{BK},\Lambda,\beta}>0$ is a constant which normalizes the positive measure into a probability measure, and it is also equal to the earlier mentioned finite volume partition function of the spherical model.

Here, we generalize the above spherical model slightly by complexifying the spin field $s_x$ and allowing for arbitrary spin-densities $\rho>0$.  Explicitly, we consider here complex fields $\phi_x\in \C$, $x\in\Lambda$,
whose values are distributed according to the measure
\begin{align}\label{eq:defmuzerointro}
 \mu_{\rho,\beta}[\rmd \phi] = \frac{1}{Z_{\rho,\beta}} \rme^{-\beta E_\Lambda[\phi]}\,
 \delta\!\left(\sum_{x\in\Lambda}|\phi_x|^2 -\rho V\right)\prod_{x\in \Lambda}\! \left[\rmd \phi^*_x \rmd \phi_x \right]\,,
\end{align}
where $\rmd \phi^*_x \rmd \phi_x := \rmd\bigl(\re\phi_x\bigr) \rmd\bigl(\im \phi_x\bigr)$.
The measure (\ref{eq:defmuzerointro}) is a ``classical field'' version of the ideal gas of bosonic particles in the canonical ensemble where the total particle number is fixed to $\rho V$ but energy is allowed to fluctuate according to the canonical Gibbs ensemble.  In fact, it follows from  our main result 
that the mechanism behind the spherical model phase transition is identical to that found for Bose--Einstein condensation of non-interacting bosons: 
if $d\ge 3$, we show that for all sufficiently large densities $\rho$ it is possible to separate a finite number of Fourier modes from the field, called the condensate, and these will carry all of the excess mass above criticality.  
The fluctuations of the remaining degrees of freedom, the normal fluid, are shown to become Gaussian
and independent from the condensate fluctuations in the large volume limit. 

An important consequence of the analysis here is to observe that the condensate cannot always be composed out
of a unique Fourier mode.  In fact, the number of relevant modes and their fluctuations might even depend on the precise shape and size of $\Lambda$.
For spin interactions, and even more so for dispersion relations arising from tight binding approximation or for 
phonons in solid state physics, it would be important to be able to consider fairly general interaction potentials.
A number of example lattice interactions are discussed in Sec.~\ref{sec:examples}.  One of these is given by a dispersion relation which has a unique global minimum but its restrictions to periodic rectangular lattices with $L$ particles on each side have a unique condensate mode for odd $L$ but $2^d$ condensate modes for even $L$.
This is in sharp contrast to the standard ideal Bose gas example \cite[Theorem 5.2.30]{bratteli:ope2}
where $L\to \infty$ limiting behaviour is unique and all excess mass condenses into the (unique) ground state, corresponding to the Fourier mode with wave number zero.

Our main result, Theorem \ref{th:mainresult}, provides explicit bounds which may be used to estimate the accuracy of any proposed splitting of the Fourier modes into condensate and normal fluid modes.  One of the main goals of the present contribution has been to find methods which would be able to identify the condensate modes properly for general, finite range
lattice interactions.  This has resulted in the bounds given in Theorem \ref{th:mainresult};
as we discuss in Sec.~\ref{sec:examples}, these bounds are indeed
sufficiently refined to distinguish the condensate modes correctly not only in the above odd and even $L$ cases, but also in all 
other examples considered in Sec.~\ref{sec:examples}.

Bose--Einstein condensation has been much more extensively studied in the literature than the spherical model.
Although the analysis is complicated by the replacement of the complex field $\phi_x$ by non-commutative 
bosonic creation and annihilation operators on the Fock space, the findings are not dissimilar from the above observations.
For example, in \cite{BuffetSP83} the properties of the condensate in the so-called 
imperfect Bose gas are shown to depend on which lattices are used to approach
the infinite volume limit, by varying the anisotropy of the lattices.  
Even more extreme examples for the ideal Bose gas are given in \cite{vdBergLewis82}.  Multi-state condensation has also been shown to occur in 
similar models in \cite{BruZagreb00} and its introduction contains a summary of other earlier findings.
In contrast, if one adds a one-particle energy gap, single-state condensation occurs for
bosons interacting via superstable two-body potentials \cite{Lauwersetal03}.   
The role the explicit gap plays in the result is discussed in the paper but, since the gap is not allowed to depend on the system size, 
it is not possible to draw conclusions about the minimal gap size needed.  Indeed, our results indicate that this dependence could be fairly complex in general.

A second motivation to study the measure (\ref{eq:defmuzerointro}) comes from statistical mechanics of discrete wave equations.
Considering $(2^{\frac{1}{2}}\re \phi_x, 2^{\frac{1}{2}}\im \phi_x)$ to form a pair of canonical variables for each $x$, one may use the function
$E_\Lambda[\phi]$ to define Hamiltonian evolution under which it is conserved and may be identified physically as the total energy.  Requiring the symmetry condition $J_{y,x}^* = J_{x,y}$ from the coupling, 
the evolution equations are equivalent to 
\[
 \partial_t \phi_x = -\ci \sum_{y\in \Lambda} J_{x,y}\phi_y\,.
\]
In particular, if $J_{x,y}=\alpha(x-y;L)$ where $\alpha$ is $L$-periodic, this corresponds to a discrete wave equation with periodic boundary conditions and with a dispersion relation $\omega$ which is given by the Fourier transform of $\alpha$.
In addition, one may check by differentiation that the $\ell_2$-norm is conserved by the time-evolution, i.e., that 
$\sum_{x\in \Lambda} |\phi_x|^2$ is also a conserved quantity.  By Liouville's theorem, the Lebesgue measure is invariant under the Hamiltonian evolution and thus the measure (\ref{eq:defmuzerointro}) yields a family of stationary measures for the discrete wave equation corresponding to the Hamiltonian $E_\Lambda[\phi]$.  Therefore, our result can also be viewed  as a proof of ``Bose--Einstein'' condensation for the equilibrium measures of these discrete wave equations.

To mention one additional motivation for the measures in (\ref{eq:defmuzerointro}), 
let us point out that they can also be obtained as a weak coupling limit of 
fixed density, i.e., ``canonical'', equilibrium measures of the discrete nonlinear Schr\"odinger equation.
In \cite{NLS09}, we study the discrete nonlinear Schr\"odinger evolution with random initial data 
distributed according to a grand canonical ensemble,
aiming at rigorous control of the related kinetic theory.
However, the assumptions used in \cite{NLS09} require that the weak coupling measure in the thermodynamic limit becomes Gaussian, hence 
excluding a range of densities which correspond to the supercritical case studied here.
The above results could provide the first step 
towards understanding kinetic theory for weakly nonlinear waves in presence of a condensate.  

The main technique for controlling the error arising from the separation of the condensate degrees of freedom is 
very different from the previous estimates
in \cite{BerlinKac52,YanWannier65}.  Instead of trying to represent the $\delta$-function in terms of oscillatory integrals,
we think of it as a constraint defining a positive measure, and aim at minimizing the effect of the separation with a flexible choice of which modes 
are included in the condensate.  It turns out that there are cases in which the condensate degrees of freedom have somewhat irregular fluctuations but the main achievement here is to show that it is possible to make the separation in such a manner that the number of condensate modes always remains bounded and the rest of the modes become independent Gaussian random variables.  After the approximate measure has been chosen, we check that it is close to the original 
one by constructing a coupling between the two measures, borrowing ideas from \cite{SaksmanWebb2016}.  
This controls the Wasserstein distance between the measures, and together with their translation invariance,
we conclude that there is a power $p'>0$ such that all finite moments of the field $\phi_x$ are $O(L^{-p'})$ close to each other as $L\to\infty$.

Couplings and Wasserstein metric are basic tools for optimal transport problems \cite{Villani:OptTransp}.  They have also been used
for studies of condensation phenomena in stochastic particle systems, although in models such as zero-range processes 
the condensation occurs at isolated lattice sites instead of Fourier modes as in the cases discussed above.  We refer to \cite{rafferty2018} 
and references therein for an up-to-date discussion and examples related to the topic.

In the following sections, we first define the complexified spherical model and describe the main results 
in more detail in Sec.~\ref{sec:mainresults}.  The fixed finite lattice case for supercritical densities is discussed in Theorem \ref{th:mainresult} while the conclusions 
for the case where a given dispersion relation is studied in the infinite volume limit are given in Corollary \ref{th:fixedomega}.
These results give bounds for the Wasserstein distance between the spherical model measure and the approximation where the condensate and normal fluid modes have been separated.  The bounds typically diverge, but in Sec.~\ref{sec:Wassersteintomoments} we explain how they nevertheless imply that the approximation errors of finite moments vanish in the infinite volume limit.  Various scenarios for the formation of the condensate for a number of example continuum dispersion relations are discussed in Sec.~\ref{sec:examples}.

In the technical part, we first prove Theorem \ref{th:mainresult} in Sec.~\ref{sec:proofofmain}, and a statement in item \ref{it:smalldev} of  Proposition \ref{th:mu1prop} which uses a number of components from the proof.  The main estimates allowing to control the infinite volume limit of fixed dispersion relations are given in Sec.~\ref{sec:proofofgap}, in particular, completing the missing proof of Lemma \ref{th:mainomlemma}.  In the two Appendices, we first
clarify the precise mathematical interpretation of the $\delta$-function constraints and recall the definition and basic properties of the Wasserstein distance.

\subsection*{Acknowledgements}

This work has greatly benefited from the input from two colleagues: Herbert Spohn, who in a personal communication suggested the splitting of the condensate fluctuations  in the infinite volume limit for the nearest neighbour interactions, similarly to what is stated 
in item \ref{it:uniquezero} of Proposition \ref{th:mu1prop}, 
and asked the question about its proof and generalizations, and Eero Saksman, who generously took the time to explain the technical details of their method for 
generating efficient couplings between two probability measures which might be mutually singular but approximate each other well with high probability.
I am also grateful to Stefan Gro\ss kinsky and Stefano Olla for their comments on coupling techniques used in stochastic particle systems.  I also thank the anonymous reviewer and Kalle Koskinen for helpful comments on the manuscript.

The work has been supported by the 
Academy of Finland via the Centre of Excellence in Analysis and Dynamics Research (projects 271983 and 307333), and it has 
also benefited from the support of the project EDNHS ANR-14-CE25-0011 of the French National Research Agency (ANR) and from the discussions during workshops organized by the Institut Henri  Poincar\'e -- Centre \'Emile Borel, Paris, France, (IHP trimester ``Stochastic Dynamics Out  of Equilibrium''), and
by the Mathematisches Forschungsinstitut Oberwolfach, Germany (MFO mini-workshop: ``Gibbs Measures for Nonlinear Dispersive Equations'').

\section{Separation of condensate in the spherical model}\label{sec:mainresults}

\subsection{Notations and definition of the spherical model measure}

We begin with the probability measure for a finite complex field $\phi_x$, $x\in \Lambda$, defined by the complexified spherical model of Berlin and Kac given in (\ref{eq:defmuzerointro}).
For simplicity, we only consider $d$-dimensional periodic lattices of fixed side length $L$, which we parameterize as follows
\begin{align}
  & \Lambda_L := \Bigl\{-\frac{L-1}{2},\ldots,\frac{L-1}{2}\Bigr\}^d \, , \qquad
\text{if $L$ is odd}\,,\\
  &\Lambda_L := \Bigl\{-\frac{L}{2}+1,\ldots,\frac{L}{2}\Bigr\}^d\, , \qquad
\text{if $L$ is even}\,.
\end{align}
Then always $V:=|\Lambda_L|=L^d$ and $\Lambda_L\subset \Lambda_{L'}$ if $L\le L'$.
Also, if $L$ is odd, $x\in\Z^d$ belongs to $\Lambda_L$ if and only if $|x|_\infty < \frac{L}{2}$.
If $L$ is even, $\Lambda_L$ contains those $x\in\Z^d$
for which $|x|_\infty \le \frac{L}{2}$ and $x_i\ne -\frac{L}{2}$ for all $i$.

We further simplify the discussion by restricting to energy functions satisfying periodic boundary conditions.  Without loss of generality, we also include the inverse temperature to the definition, and thus assume that
\[
 \beta E_\Lambda[\phi] =  H_L[\phi] := \sum_{x,y\in \Lambda_L} \phi_x^* \alpha(x-y;L) \phi_y \,,
\]
where $\alpha:\Lambda_L\to \C$ determines the interaction energies.  Here, and in the following,
we use periodic arithmetic on $\Lambda_L$, setting 
$x'\pm x := (x'{\pm}x) \bmod \Lambda_L$ and $-x:= (-x) \bmod \Lambda_L$,
for $x',x\in  \Lambda_L$. 

The above definition implies that the energies remain invariant under periodic translations
of the field configuration, i.e., $H_L[\phi']=H_L[\phi]$ if $y\in \Lambda_L$ and $\phi'_x:=\phi_{x+y}$, $x\in \Lambda_L$.
In fact, we can now ``diagonalize'' the interaction by using discrete Fourier transform.
We define the Fourier transform on $\Lambda=\Lambda_L$ by first setting as
the dual lattice $\Lambda^*(L) := \Lambda_L/L\subset \left]{-}\frac{1}{2},\frac{1}{2}\right]^d$ and then 
denoting the Fourier transform of a function
$f:\Lambda\to\mathbb{C}$ by $\FT{f} : \Lambda^* \to
\mathbb{C}$, where
\begin{equation}
   \FT{f}(k) = \sum_{x\in \Lambda} f(x) \textrm{e}^{-\textrm{i} 2\pi k \cdot x}\,, \qquad k\in \Lambda^*\,.
\end{equation}
The inverse transform is given by
\begin{equation}
   \IFT{g}(x) = \frac{1}{V} \sum_{k\in \Lambda^*}  g(k)
     \textrm{e}^{\textrm{i} 2\pi k \cdot x} =: \int_{\Lambda^*}\!\rmd k\,  g(k)
     \textrm{e}^{\textrm{i} 2\pi k \cdot x}\,, \qquad x\in \Lambda\,.
\end{equation}
It is straightforward to check that both transforms are pointwise invertible for all $f$ and $g$, 
$f(x)=\IFT{(\FT{f})}(x)$ for $x\in \Lambda$ and $g(k) = \FT{(\IFT{g})}(k)$ for $k\in \Lambda^*$.

The standard convolution results hold for the discrete Fourier transform, and thus we have 
\[
 H_L[\phi] = \int_{\Lambda^*}\!\rmd k\, \omega(k) |\Phi_k|^2 =: H[\Phi] \,,
\]
where $\Phi = \FT{\phi}:\Lambda^*\to \C$ and $\omega = \FT{\alpha}$.  In this formulation, it is now obvious that
if we wish to satisfy the physical requirement of the energy $H_L$ being real for all field configurations, it is necessary that $\omega(k)\in\R$ for all $k\in \Lambda^*$.
In addition, by the inversion formula
\begin{align}\label{eq:defalphaL}
 \alpha(x;L) :=  \int_{\Lambda^*}\!\rmd k\,  \omega(k)
     \textrm{e}^{\textrm{i} 2\pi k \cdot x}\,. 
\end{align}
Therefore, it is possible to simplify the study of the infinite volume limit $L\to \infty$
by considering a ``target'' function $\omega:\T^d\to \R$, parameterizing the torus using 
$\left]{-}\frac{1}{2},\frac{1}{2}\right]^d$, and defining $\alpha$ using the formula (\ref{eq:defalphaL}).
For reasons explained in the Introduction, we call such functions $\omega$ \defem{dispersion relations}.
In the following, some of the results concern the limiting behaviour as $L\to \infty$ for some given dispersion relation $\omega$ on the torus, 
while others assume that $L$ is fixed and $\omega(k)$, $k\in \Lambda^*$, are some given real numbers.

We also denote
 \[
 N_L[\phi] = \sum_{x\in \Lambda_L} |\phi_x|^2 \,,
\]
and thus arrive at the following expression for the spherical model measure
\begin{align}\label{eq:defmuzero}
 \mu_{\rho,\beta}[\rmd \phi] = \frac{1}{Z_{\rho,\beta}} \rme^{-H_L[\phi]}\,
 \delta\!\left(N_L[\phi] -\rho V\right)\prod_{x\in \Lambda}\! \left[\rmd \phi^*_x \rmd \phi_x \right]\,.
\end{align}
By the discrete Plancherel theorem, here $N_L[\phi]=\norm{\phi}^2=\norm{\Phi}^2=: N[\Phi]$, and we observed earlier 
that $H_L[\phi]=H[\Phi]$. Since the Fourier transform introduces an invertible linear transformation of the field, we may conclude that the spherical model measure has a particularly simple form for the Fourier components $\Phi_k = \FT{\phi}_k$ of the field, 
\begin{align}\label{eq:defmu0}
 \mu_0[\rmd \Phi] := \frac{1}{Z_{\rho}} \rme^{-H[\Phi]} \delta(N[\Phi]- \rho V)
 \prod_{k\in \Lambda^*} \left[\rmd \Phi^*_k \rmd \Phi_k\right]
\end{align}
where $\rmd \Phi^*_k \rmd \Phi_k := \rmd\bigl(\re\Phi_k\bigr) \rmd\bigl(\im \Phi_k\bigr)$, $Z_\rho$ normalizes the integral to one, 
and
\[
 H[\Phi] := \int_{\Lambda^*}\!\rmd k\, \omega(k) |\Phi_k|^2\,,\qquad
 N[\Phi] := \int_{\Lambda^*}\!\rmd k\, |\Phi_k|^2 \,.
\]

As the norm in which to measure the Wasserstein distance, we choose the $\ell_2$-metric on the $x$-space.
By the Plancherel theorem for discrete Fourier transform, this means using the following norm for the field $\Phi_k$,
\[
 \norm{\Phi}^2 := \int_{\Lambda^*}\!\rmd k\, |\Phi_k|^2\,,
\]
and $N[\Phi]=\norm{\Phi}^2$.
We also need spherical coordinates in these variables.  
We denote the radial distance coordinate by $|\Phi|$, and it is then related to the above norm by 
\[
 |\Phi|^2 := \sum_{k\in \Lambda^*} |\Phi_k|^2 = |\Lambda|\, \norm{\Phi}^2\,.
\]

\subsection{Factorized supercritical measures}\label{sec:factsupercrit}

Our goal is to study the spherical model for parameter values which 
lead to generation of a condensate.  Since this is a physical, macroscopic notion, we first need to quantify mathematically what it could mean for \defem{finite} 
lattice systems such as the spherical model measure introduced in the previous subsection.
After this, 
we will separately consider the large $L$ behaviour of systems whose energy eigenvalues $\omega(k)$, $k\in \Lambda^*$, 
arise from a continuum dispersion relation $\omega:\T^d\to \R$ as explained earlier.

To quantify condensates and supercriticality, it will be necessary to identify a sufficiently large energy gap separating the 
modes which belong to the condensate from the rest.
To this end, we divide
the wave numbers in $\Lambda^*$ into a \defem{condensate wave number set} $\Lambda_0^*$ and a \defem{normal fluid
wave number set} $\Lambda^*_+ =\Lambda^*\setminus \Lambda_0^*$ in such a manner
that the energies occurring in these sets are separated by a non-empty interval.
An important parameter of the split turns out to be the proportional size of the gap, after normalizing the lowest energy to zero;
the following item collects the related definitions and terminology.
\begin{definition}\label{th:defslpit}
 Consider $\Lambda^*$ for some fixed $L$ and suppose $\omega(k)\in \R$, $k\in \Lambda^*$, are  given.   Define $\omega_0:=\min_{k\in\Lambda^*} \omega(k)$
 and $e_k:=\omega(k)-\omega_0\ge 0$, $k\in\Lambda^*$.
 A \defem{split} of $\Lambda^*$ is a pair $(\Lambda_0^*,\Lambda_+^*)$ of nonempty disjoint subsets of $\Lambda^*$ whose union covers the whole $\Lambda^*$.
 Given $0\le a<b$ and a split $(\Lambda_0^*,\Lambda_+^*)$, we say that the \defem{split is separated by the energy interval $[a,b]$} if 
 $e_k \le a$ for all $k\in \Lambda^*_0$ and $e_k \ge b$ for all $k\in \Lambda^*_+$.  In this case, the \defem{relative energy gap of the split}
 is defined as $\delta^{-1}$ where 
 \[
   \delta := \frac{\max_{k\in \Lambda^*_0}e_k}{\min_{k\in \Lambda^*_+}e_k}\le \frac{a}{b}<1\,.
 \]
 We denote the number of elements in the two subsets of the split by 
 $V_0:=|\Lambda_0^*|$ and $V_+:=|\Lambda_+^*|$.
\end{definition}
Since $V=|\Lambda^*|$, for such a split we clearly have $0<V_0,V_+<V$ and
$V=V_0+V_+$.  Also, every global lattice minima, a point $k\in \Lambda^*$ at which $\omega(k)=\omega_0$,
belongs to $\Lambda^*_0$.  Hence, $\Lambda^*_0$ contains all $k$ for which $e_k=0$,
and thus $e_k>0$ for all $k\in \Lambda^*_+$.

Given such a split, we call the field $\Phi^+$ 
composed out of modes with $k\in \Lambda_+$ the \defem{normal fluid} while
the field $\Phi^0$ resulting from the remaining modes is called the \defem{condensate}.
The goal is to quantify under which assumptions the condensate field can be composed
out of a small fraction of the modes, $\frac{V_0}{V}\ll 1$, so that 
they nevertheless carry a substantial fraction of the total mass $\rho V$.  Analogously
to the Bose--Eistein condensation, one could then expect the 
normal fluid to fluctuate according to the critical thermal, grand canonical ensemble.
Indeed, under the assumptions made in the main theorem we can prove that the 
normal fluid $\Phi^+$ follows very accurately Gaussian statistics given by the following
distribution
\begin{align}\label{eq:defmuphiplus}
 & \mu_+[\rmd \Phi] := \frac{1}{Z_+} 
 \prod_{k\in \Lambda_+^*} \left[\rmd \Phi^*_k \rmd \Phi_k\right]  \rme^{-L^{-d} \sum_{k\in \Lambda_+^*}
 (\omega(k)-\omega_0) |\Phi_k|^2}\,. 
\end{align}
This measure is well-defined since $\omega(k)>\omega_0$ for all $k\in \Lambda^*_+$.
The expectation of norm density, $\mean{\norm{\Phi_+}^2/V}$, under such a measure is equal to
\begin{align}\label{eq:defrhocrit}
 \rhoc(L) := \int_{\Lambda_+^*}\!\!\rmd k\, \frac{1}{\omega(k)-\omega_0}\,.
\end{align}
The standard deviation of the norm density is proportional to $1/\sqrt{V}=L^{-\frac{d}{2}}$, 
and thus for large $L$ the normal fluid under this measure cannot carry much more
of the density fixed by the condition $N[\Phi]=\rho V$ as soon as $\rho>\rhoc$.
Since $N[\Phi]= \norm{\Phi_+}^2+ \norm{\Phi_0}^2$, then the extra norm density
$\rho-\rhoc$ will be contained in the condensate modes.  

Based on the above analogy, we say the the spherical model is \defem{supercritical}
if $\rho>\rhoc$ for a split which has sufficiently large relative energy gap and only a few condensate modes (the precise conditions are given in Theorem \ref{th:mainresult}).
The above formal discussion will then turn out to give the correct picture for fairly general energy functions $\omega(k)$.  In fact, the separation between the two sets of modes is so strong that even the fluctuations of
condensate and of the normal fluid will become statistically independent.
However, if the condensate is degenerate, the fluctuations of the condensate can be nontrivial.  In the main result we will compare the spherical model measure $\mu_0$
to the probability measure $\mu_1$ defined by
 \begin{align}\label{eq:defmu1}
 & \mu_1[\rmd \Phi] := \frac{1}{Z_1} 
 \prod_{k\in \Lambda_+^*} \left[\rmd \Phi^*_k \rmd \Phi_k\right]  \rme^{-E_+[\Phi]} 
 \nonumber \\ & \qquad \times
 \prod_{k\in \Lambda_0^*} \left[\rmd \Phi^*_k \rmd \Phi_k\right] \rme^{-E_0[\Phi]\left(1-\frac{\rhoc}{\Delta}\right)}
 \prod_{k\in \Lambda^*_+} \left(1-\frac{E_0[\Phi]L^{-d}}{e_{k} \Delta} \right)^{-1}
  \delta(\rho_0[\Phi] - \Delta) \,,
\end{align}
where $\Delta:=\rho-\rhoc>0$, $Z_1$ is a constant normalizing the integral to one and, using $e_k:=\omega(k)-\omega_0$, we define
\begin{align}\label{eq:defmu1parameters}
\rho_0[\Phi]:= \frac{1}{V}\int_{\Lambda_0^*}\!\rmd k\, |\Phi_k|^2\,,\quad  
 E_+[\Phi] := \int_{\Lambda_+^*}\!\rmd k\, e_k |\Phi_k|^2\,,\quad 
 E_0[\Phi] := \int_{\Lambda_0^*}\!\rmd k\, e_k |\Phi_k|^2
 \,.  
 \end{align}
Clearly, $\mu_1$ is a product of $\mu_+$ and a measure for the condensate modes, and the total norm density is split between the normal fluid and condensate in the manner described above.

The structure of the condensate fluctuations under $\mu_1$ may indeed be fairly complicated.
However, there are certain situations where they can be replaced by simpler uniform 
distribution of the excess mass over the condensate modes, i.e., by using the measure 
 \begin{align}\label{eq:defmu1prime}
 & \mu'_1[\rmd \Phi] := \frac{1}{Z'_1} 
 \prod_{k\in \Lambda_+^*} \left[\rmd \Phi^*_k \rmd \Phi_k\right]  \rme^{-E_+[\Phi]} 
  \prod_{k\in \Lambda_0^*} \left[\rmd \Phi^*_k \rmd \Phi_k\right]\delta(\rho_0[\Phi] - \Delta) \,.
 \end{align}
instead of $\mu_1$ above.  Some sufficient conditions for using 
the simpler measure are discussed later in Proposition \ref{th:mu1prop}.  
As we show there, using $\mu'_1$ is allowed at least if a single mode condensate can be used, i.e., if $V_0=1$.
We call both $\mu_1$ and $\mu'_1$ 
\defem{factorized supercritical measures}.

\subsection{Main results}

Our main result is to state conditions under which $\mu_0$ and $\mu_1$ are so close to each other that the expectations of all local observables will agree 
with each other, up to some error which is proportional to a negative power of $L$, hence vanishing when $L\to \infty$.  
The precise conditions are contained in the following Theorem implying a bound for the Wasserstein distance between $\mu_0$ and $\mu_1$.  The proof of Theorem is given in Section \ref{sec:proofofmain}.

The Wasserstein distance estimate is sufficiently strong that 
local expectations of the original field, $\phi = \tilde{\Phi}$, generated by these two measures agree up to errors
which vanish as $L\to \infty$.  Namely, if $I\subset \Lambda$ is finite in the sense that $|I|/V\ll 1$ and $\phi^I:=\prod_{x\in I}\phi_x$, then the bound given in Theorem \ref{th:mainresult} implies the existence of $p'>0$ such that
\[
 \mean{\phi^I}_{\mu_0} = \mean{\phi^I}_{\mu_1} + O(L^{-p'})\,.
\]
The proof of this statement will rely on translation invariance of the random field generated by 
the measures $\mu_0$ and $\mu_1$ and it is given later as Theorem \ref{th:infvollocalconv} in Sec.~\ref{sec:Wassersteintomoments}.
Therefore, if a split with sufficiently large gap can be found, then the spherical model is well approximated by a critical Gaussian field and a few independent condensate Fourier modes, as determined by $\mu_1$.  

\begin{theorem}\label{th:mainresult}
 Consider a fixed $L$ and some given $\omega(k)\in \R$, $k\in \Lambda^*$.
 Suppose $(\Lambda_0^*,\Lambda^*_+)$ is a split of $\Lambda^*$ which is
 separated by the energy interval $[a_L,b_L]$, $0\le a_L<b_L$, and has a relative energy gap $\delta_L^{-1}$, as specified in Definition \ref{th:defslpit}.  We recall also the definitions
 of the total system size $V$,  the number of the condensate modes $V_0$,
 and the critical norm density $\rhoc(L)$ in (\ref{eq:defrhocrit}).
 
 Define the measure $\mu_0$ by (\ref{eq:defmu0}) and suppose that it is supercritical in
 the sense that $\rho>\rhoc$. Denote $\Delta:=\rho-\rhoc(L)$,
 and assume that the gap and lattice size are large enough so that 
 \begin{align}\label{eq:deltavepassump}
  \delta_L \le \frac{1}{2}\,, \qquad
  \vep_L := \max\left(2 \delta_L,\frac{1}{V^2\rhoc^2 } \sum_{k\in \Lambda^*_+} \frac{1}{e_k^2}\right)
  \le \frac{\Delta^2}{2^5 V_0^2\rho^2}\,.   
 \end{align}
 Define the measure $\mu_1$ by (\ref{eq:defmu1}).
 
Then there exists a constant $C_2>0$
such that the $2$-Wasserstein distance between $\mu_0$ and $\mu_1$ satisfies
\begin{align}\label{eq:mainW2bound}
  W_2(\mu_0,\mu_1) \le C_2 L^{\frac{d}{2}} \vep_L^{\frac{1}{4}}\,.  
\end{align}
In particular, the inequality 
holds with the choice $C_2= 2^4(\rho/\Delta)^{V_0/2} \sqrt{ (\rho+\Delta) V_0}$.
\end{theorem}

As shown later in Lemma \ref{th:mainomlemma}, for energies arising from many common continuum dispersion relations
a sequence of splits can be found for which $\vep_L\to 0$ as $L\to \infty$ while $V_0$ and $\rhoc(L)$ remain bounded, implying $C_2=O(1)$ if $\rho > \sup_L \rhoc(L)$.
However, the speed of convergence of $\vep_L$  is usually not sufficient for the bound of the Wasserstein distance $W_2(\mu_0,\mu_1)$ to go to zero, 
so we cannot state any convergence result in the above (unscaled) $L^2$-norm.  
Nevertheless, as we show in Sec.~\ref{sec:Wassersteintomoments},
for errors in local correlation functions the bound can be improved by a factor of $L^{-\frac{d}{2}}$
which shows that these errors vanish in the limit of large lattices. The precise statement is given in Theorem \ref{th:infvollocalconv}, 
and as discussed in Sec.~\ref{sec:Wassersteintomoments},
the main simplification from the replacement of $\mu_0$ by $\mu_1$ is given by the vastly simpler fluctuation properties of the normal fluid under 
the measure $\mu_1$.  

There are a few special cases for which also the condensate fluctuations have simple structure, summarized in Proposition \ref{th:mu1prop}.  
In the statements below, we say for instance that ``$\Phi=\Phi^+ + L^d \sqrt{\Delta} X$ in distribution, where   $X$ is a random variable independent of $\Phi^+$ and uniformly distributed on the unit sphere $S^{2 V_0-1}$''.  There
it is implicitly assumed that the first term refers to normal fluid components and the second to the condensate components using the standard isomorphism between $\C^{\Lambda^*_0}$ and $\R^{2 V_0}$: 
for $k\in \Lambda^*_+$, we then have $\Phi_k=\Phi^+_k$, and for $k\in \Lambda^*_0$, we have $\Phi_k=L^d \sqrt{\Delta}
(X_{2 p(k) -1} + \ci X_{2 p(k)})$ where $p:\Lambda^*_0\to\set{1,2,\ldots,V_0}$ is any bijection, i.e., some enumeration of $\Lambda^*_0$.  (Since the uniform measure on the unit sphere $S^{d-1}$ is invariant under permutation of the $d$ 
coordinate labels, the distribution does not depend on the choice of the enumeration $p$.)

\begin{proposition}\label{th:mu1prop}
Suppose that all the assumptions and definitions in Theorem \ref{th:mainresult} hold, in particular, we recall Definition \ref{th:defslpit}.
Let $\Phi^+$ denote the Gaussian lattice field distributed according to the measure $\mu_+$ defined in (\ref{eq:defmuphiplus}).
 \begin{enumerate}
  \item\label{it:uniquezero} If $V_0=1$, then $\Phi=\Phi^+ + L^d\sqrt{\Delta} \rme^{\ci \theta}$ in distribution, where 
  $\theta$ is a random variable independent of $\Phi^+$ and uniformly distributed on the interval $[0,2\pi]$.
  \item\label{it:allground} If $\omega(k)$ is a constant for $k\in\Lambda_0^*$, then in distribution $\Phi=\Phi^+ + L^d \sqrt{\Delta} X$, where   
  $X$ is a random variable independent of $\Phi^+$ and uniformly distributed on the unit sphere $S^{2 V_0-1}$.
  \item\label{it:smalldev} If there is a non-negative $\tilde{\vep}\le 1$ such that
  $e_k\le \frac{1}{2 \rho}L^{-d}\tilde{\vep}$ for $k\in \Lambda_0^*$, then 
  \begin{align}\label{eq:W2primebound}
  W_2(\mu_0,\mu'_1) \le L^{\frac{d}{2}} 2^4 \sqrt{(\rho + \Delta)V_0} \left((\rho/\Delta)^{V_0/2} \vep_L^{\frac{1}{4}}+  \tilde{\vep}^{\frac{1}{2}}\right)  
\end{align}
for the measure $\mu'_1$ defined in (\ref{eq:defmu1prime}).
 Under the measure $\mu'_1$ we have
  $\Phi=\Phi^+ + L^d\sqrt{\Delta} X$ in distribution, where $X$ is a random variable independent of $\Phi^+$ and uniformly distributed on the unit sphere $S^{2 V_0-1}$.
 \end{enumerate}
\end{proposition}
\begin{proof}
 The assumptions in the first two items imply that $E_0[\Phi]=0$ (note that by definition of the split, we necessarily have $\omega(k)=\omega_0$ for some, and hence for all,  $k\in \Lambda^*_0$). Thus the weight related to $k\in \Lambda_0^*$
 is equal to one.  Since $\rho_0[\Phi] = V^{-2} |\Phi^0|^2$, where $|\Phi^0|$ denotes the Euclidean norm in 
 $\C^{V_0}\cong\R^{2 V_0}$,
 the random variable $X:= (L^{-d}\Delta^{-\frac{1}{2}}\re \Phi_k^0,L^{-d}\Delta^{-\frac{1}{2}}\im \Phi_k^0)_{k\in \Lambda^*_0}$ is uniformly distributed on the unit sphere $S^{2 V_0-1}$: for any continuous bounded function $f:\R^{2 d}\to \C$ we have in spherical coordinates
 \begin{align*}
& \int \prod_{k\in \Lambda_0^*} \left[\rmd \Phi^*_k \rmd \Phi_k\right]
  \delta(\rho_0[\Phi] - \Delta) f(\Phi) = 
  \int_{\R^{2 V_0}}\!  \rmd^{2 V_0\!}X\, 
  \delta(\Delta (|X|^2-1)) f(V \sqrt{\Delta} X) 
  \\ & \quad
  = 
  \frac{1}{\Delta} \int_{S^{2 V_0-1}}\!\rmd \Omega \int_0^\infty\! \rmd r\, r^{2 V_0-1}
  \delta(r^2-1) f(V \sqrt{\Delta} r\Omega) 
  \\ & \quad
  =  
  \frac{1}{2\Delta} \int_{S^{2 V_0-1}}\!\rmd \Omega \int_0^\infty\! \rmd s\, s^{V_0-1}
  \delta(s-1) f(V \sqrt{\Delta} \sqrt{s}\Omega)
  =  
  \frac{1}{2\Delta} \int_{S^{2 V_0-1}}\!\rmd \Omega\,f(V \sqrt{\Delta}\Omega)
 \end{align*}
 and the normalization condition fixes the overall constant correctly.
 
 If $V_0=1$, $X$ is uniformly distributed on the unit circle and thus equals $\rme^{\ci \theta}$ in distribution.
 The proof of the last item uses techniques from the proof of the main Theorem, and it can be found at the end of Section \ref{sec:proofofmain}.
\end{proof}

To study infinite volume limits, we assume that the weights $\omega(k)$ are given by an $L$-independent dispersion relation, satisfying the following conditions.
\begin{assumption}\label{as:omega}
Suppose $d\ge 3$ and 
 consider a function $\omega:\T^d\to \R$ which is $C^2$ and has only finitely many non-degenerate minima.  More precisely, we assume that both of the following statements hold:
 \begin{enumerate}
  \item The periodic extension of $\omega$ into a function $\R^d\to \R$ is twice continuously differentiable.
  \item By the first assumption and compactness of $\T^d$, $\omega$ attains a minimum value $\ommin \in \R$.  
  We assume that the collection of all global minima in $\T^d$, $T_0:=\defset{k\in \T^d}{\omega(k)=\ommin}$, is finite and that 
  the Hessian matrix $D^2 \omega(k_0)$ is invertible for all  $k_0\in T_0$
 \end{enumerate}
\end{assumption}
Note that these assumptions are invariant if $\omega$ is multiplied by any positive constant, and thus they remain invariant 
in changes of the implicit inverse temperature factor $\beta$.

It turns out that in the presence of a condensate, the distribution around the degrees 
of freedom with minimum energy may vary with the lattice size $L$ without converging towards any
limiting behaviour as $L\to \infty$.  For example, in Section \ref{sec:dispwithseveralmin}
we present an example with different number of condensate modes for odd and even $L$.
We illustrate via explicit examples why the split can have nontrivial dependence on the lattice size $L$ in Section \ref{sec:examples}.

The following Lemma shows that for dispersion relations satisfying Assumption \ref{as:omega} a split with the desired properties can be found.
\begin{lemma}\label{th:mainomlemma}
 Suppose that $d\ge 3$ and $\omega$ satisfies Assumption \ref{as:omega}.  For each $L$, define $\omega_0$
 and $e_k$, $k\in\Lambda^*$, as in Definition \ref{th:defslpit}.
 Choose $\kappa$ such that $0<\kappa<\frac{d}{2}$, if $d\ge 4$, and $0<\kappa < 1$, if $d=3$.
 Then there are constants $L_0,M_0\in \N_+$ and $c_0,c_2>0$, depending only on $d$, the function $\omega$, and the choice of $\kappa$, such that for all $L\ge L_0$ 
 we can find a split $(\Lambda_0^*,\Lambda^*_+)$ of $\Lambda^*$ with the following properties:
 \begin{enumerate}
  \item\label{it:splitcondensate} $M_0$ can be chosen independently of $\kappa$, $|\Lambda_0^*|\le M_0$, and for every $k\in \Lambda_0^*$,
  \begin{align}\label{eq:maingap}
  0\le \omega(k)-\ommin < c_0 L^{-2}\,.
 \end{align}
  \item\label{it:splitgap} The split is separated by an energy interval $[a_L,b_L]$ and has a relative energy gap $\delta_L^{-1}$, where
  $b_L\ge \frac{1}{2}c_0 L^{-d+\kappa}$ and
\begin{align}\label{eq:deltaLbound}
 \delta_L\le L^{-\frac{d-2-\kappa}{M_0}}\le 1\,.
\end{align}  
 \item\label{it:splitbounds} We have
 \begin{align}\label{eq:e2sumbound}
   \frac{1}{V^2} \sum_{k\in \Lambda^*_+} \frac{1}{e_k^2} \le c_2 L^{-2\kappa}\,,
 \end{align}
 the following positive integral is finite,
 \begin{align}\label{eq:defrhoinfty}
 \rho_\infty := \int_{\T^d}\!\rmd k\, \frac{1}{\omega(k)-\ommin} < \infty\,,  
 \end{align}
 and, as $L\to\infty$, 
\begin{align}\label{eq:rhoerrorbound}
 \rhoc(L) = \rho_\infty + O(L^{-\min(\kappa,2)})\,.
\end{align}
\end{enumerate}
In particular,  $\max_{k\in \Lambda^*_0} \omega(k)\to \ommin$,
$\rhoc(L)\to \rho_\infty$, and $\delta_L \to 0$, as $L\to \infty$.
\end{lemma}

The proof of the Lemma is postponed to Sec.\ \ref{sec:proofofgap}, and it contains ways to construct some constants
for which the Theorem holds.  However, these constructions are not always optimal
since they need to take into account extreme cases such as very anisotropic dispersion relations.
Hence, if optimal decay estimates are desired, it is better to optimise the values case by case instead of using, e.g., the worst case 
estimate in (\ref{eq:defmaxN0}) for $M_0$.

As a straightforward application, we obtain the following consequences for systems where the infinite lattice dispersion relation is kept fixed and $L$ is taken large.
\begin{corollary}\label{th:fixedomega}
 Suppose that $d\ge 3$ and $\omega$ satisfies Assumption \ref{as:omega}, and take some 
 cutoff parameters for the minimum distance from criticality, $\Delta_0>0$, and for a maximal density, $\bar{\rho}> \rho_\infty + \Delta_0$, 
 where $\rho_\infty$ is defined by (\ref{eq:defrhoinfty}).
 
 Then there are $L'$, $M_0$, and $C'>0$ such that for any $L\ge L'$ 
 we can find a split $(\Lambda_0^*,\Lambda^*_+)$ of $\Lambda^*$ satisfying all properties stated in Lemma \ref{th:mainomlemma} and for which the Wasserstein
 distance between the measures $\mu_0$ and $\mu_1$ defined in Theorem \ref{th:mainresult} satisfies
 \begin{align}\label{eq:largeLW2bound}
  W_2(\mu_0,\mu_1) \le C' L^{\frac{d}{2}-\frac{d/2-1}{2 M_0+1}}\,,
\end{align}
 for all densities $\rho$ on the interval
 \begin{align}\label{eq:rhocondition}
  \sup_{L\ge L'} \rhoc(L) + \Delta_0\le \rho \le \bar{\rho} \,.  
 \end{align}
\end{corollary}
\begin{proof}
 Since the assumptions of Lemma \ref{th:mainomlemma} are satisfied, 
 $\rhoc(L)\to \rho_\infty$, as $L\to \infty$, and thus there is $L_0'$ such that 
 $\sup_{L\ge L'_0} \rhoc(L) + \Delta_0<\bar{\rho}$.  Therefore, if $L'\ge L_0'$, there are densities $\rho$
 for which (\ref{eq:rhocondition}) holds.

 In addition, we can conclude from the Lemma that 
 there is $M_0\ge 1$ such that for any appropriately chosen $\kappa$, 
 the split $(\Lambda_0^*,\Lambda^*_+)$ of $\Lambda^*$ obtained from the Lemma
 satisfies
 $\delta_L\le L^{-\frac{d-2-\kappa}{M_0}}$ and $\vep_L=O(\delta_L+L^{-2\kappa})$.
 Thus both go to zero as $L\to \infty$.  Now if $L'\ge \max(L_0, L_0')$, $L\ge L'$, and $\rho$ satisfies (\ref{eq:rhocondition}), we have 
$\Delta_0\le \Delta\le \bar{\rho}$ and
$\frac{\Delta^2}{2^5 V_0^2\rho^2}\ge \frac{\Delta_0^2}{2^5 M_0^2\bar{\rho}^2}>0$, uniformly in $L$.
 Therefore, we may find $L'\ge \max(L_0, L_0')$ such that both inequalities in (\ref{eq:deltavepassump}) hold for all $L\ge L'$ and all $\rho$ satisfying (\ref{eq:rhocondition}).
 
 Thus we may use the conclusions of the main Theorem for these values of parameters, and the constant $C'=C_2$ may be adjusted to work for all allowed values of $\kappa$, $L$, and $\rho$.  
 Since also $M_0$ is independent of $\kappa$, we can maximize the decay of $\vep_L$ by setting $\kappa=\frac{d-2}{2 M_0+1}<\frac{d}{2}$ which satisfies $\kappa<1$ for $d=3$.  This results in the bound stated in the Corollary.
\end{proof}

\section{Local correlation estimates from Wasserstein bounds}\label{sec:Wassersteintomoments}

In the main result, a bound is derived for the Wasserstein distance between two measures $\mu_0$ and $\mu_1$
which are both \defem{gauge invariant} in the sense that $(\Phi_k)_{k\in \Lambda^*}$ and $(\rme^{\ci \varphi_k} \Phi_k)_{k\in \Lambda^*}$
have the same distribution for any choice of the constant phase shifts $\varphi_k\in \R$, $k\in \Lambda^*$.  This is a consequence of the geometric identification 
between $\C$ and $\R^2$ which implies that a multiplication
$\Phi_k \to \rme^{\ci \varphi_k}\Phi_k$ corresponds to a rotation by an angle $\varphi_k$ and thus it leaves the Lebesgue measure
$\rmd (\re \Phi_k)\, \rmd (\im \Phi_k)$ invariant.  The weight functions only depend on $|\Phi_k|^2$ and thus also they are left invariant.

However, in applications, one is usually mainly interested in the corresponding fields $\phi_x$, $x\in \Lambda_L$, obtained by inverse Fourier transform from 
$\Phi_k$: we consider the collection of 
\begin{align}\label{eq:defphix}
 \phi_x = \int_{\Lambda^*}\!\rmd k\, \Phi_k \rme^{\ci 2\pi k \cdot x}\,, 
\end{align}
for $x\in \Lambda_L$.  The above gauge invariance of the Fourier components is reflected in \defem{translation invariance} of the field $\phi_x$.
Namely, for any $y\in \Lambda_L$, we have
\[
 \phi_{x+y} = \int_{\Lambda^*}\!\rmd k\, \rme^{\ci 2\pi k \cdot x}  \rme^{\ci 2\pi k \cdot y} \Phi_k \,, \quad  x\in \Lambda_L\,,
\]
and thus the field $(\phi_{x+y})_{x\in \Lambda}$ has the same distribution as the field $(\phi_{x})_{x\in \Lambda}$.

This translation invariance is sufficient to lift the earlier usually divergent Wasserstein bounds to vanishing error estimates for moments of the 
field $\phi_x$.  To see this, consider a sequence $I$ of length $n\ge 1$ of pairs $(x_i,\tau_i)_{i=1}^n$, where 
$x_i\in \Lambda_L$ and $\tau_i\in \set{-1,1}$.  We use the index $\tau$ to determine complex conjugation: we set $\phi_{x,1}=\phi_x$ and $\phi_{x,-1}=\phi_x^*$, and use the shorthand notation 
${\phi}^I :=  \prod_{\alpha\in I} {\phi}_\alpha := \prod_{i=1}^n {\phi}_{x_i,\tau_i}$ 
for the monomial corresponding to the above sequence $I$.
The expectation of such local observables will get an improvement by a factor $L^{-\frac{d}{2}}$ for the Wasserstein distance from translation invariance,
as stated in the following Lemma.
\begin{lemma}\label{th:localobsestimate}
Suppose $\mu$ and $\mu'$ are gauge invariant measures for the Fourier components, field $\Phi(k)$, $k\in \Lambda^*$.
Given $x\in \Lambda$, define $A_1(x):=1$ and for $n>1$ set
\begin{align}\label{eq:defAnmom}
 A_n(x) := \max\left(\mean{|\phi_{x}|^{2(n-1)}}^{(2(n-1))^{-1}}_{\mu},\mean{|\phi_{x}|^{2(n-1)}}^{(2(n-1))^{-1}}_{\mu'}\right)\, .
\end{align}
Consider the random field $\phi= \tilde{\Phi}$ and
suppose $n\ge 1$ is such that $A_n(x)<\infty$ for some $x\in \Lambda$.

Then $A_n(x)$ does not depend on the choice of $x$ and for any sequence $I$ of length $n$ as above, we have an estimate
\begin{align}\label{eq:finmomestimate}
 \left|\mean{\phi^I}_{\mu} - \mean{\phi^I}_{\mu'}  \right| \le 
   A_n^{n-1} n W_2(\mu,\mu')L^{-d/2}\, . 
\end{align}
\end{lemma}
\begin{proof}
Under either of the measures $\mu$ and $\mu'$ the field $\phi_x$ is translation invariant, $\mean{\phi^I}=\mean{\phi^{I+y}}$
for any $y\in \Lambda_L$, where $I+y:=((x_i+y,\tau_i))_{i=1}^n$.
Therefore, for any coupling $\gamma$ between $\mu$ and $\mu'$ the difference of their moments satisfies
\begin{align}\label{eq:defdiffX}
 X := \mean{\phi^I}_{\mu} - \mean{\phi^I}_{\mu'} 
 = \frac{1}{V} \sum_{y\in \Lambda_L}\mean{\phi^{I+y}}_{\mu} -\frac{1}{V} \sum_{y\in \Lambda_L} \mean{\phi^{I+y}}_{\mu'} 
 = \frac{1}{V} \sum_{y\in \Lambda_L} \mean{\phi^{I+y} - (\phi')^{I+y}}_{\gamma}\,. 
\end{align}
In particular, if $n=1$, by using Cauchy--Schwarz inequality we obtain
\[
 |X|\le \frac{1}{V}
 \sum_{y\in \Lambda}
  \left\langle |\phi_{x_1+y}-\phi'_{x_1+y}|^2\right\rangle_{\gamma}^{\frac{1}{2}} 
 \le \frac{1}{V} \sqrt{V}
  \mean{\norm{\phi-\phi'}_2^2}_\gamma^{\frac{1}{2}} 
 =  L^{-\frac{d}{2}} A_n^{n-1} n
 \mean{\norm{\Phi-\Phi'}^2}_\gamma^{\frac{1}{2}} \,.
\]
Since the left hand side does not depend on the coupling $\gamma$, taking an infimum yields 
the bound in (\ref{eq:finmomestimate}); cf.\ the definition of the Wasserstein distance in Appendix \ref{sec:Wdistance}.

Consider then the case $n>1$.
The difference of products in (\ref{eq:defdiffX}) can be ``telescoped'' as follows
\[
 \prod_{i=1}^n \phi_i = \prod_{i=1}^n \phi_i' +\sum_{i=1}^n (\phi_i-\phi_i') \prod_{j=1}^{i-1} \phi_j \prod_{j=i+1}^{n} \phi'_j \,,
\]
yielding an estimate
\[
 \left| \phi^{I+y} - (\phi')^{I+y}  \right|
 \le \sum_{i=1}^n 
 |\phi_{x_i+y}-\phi'_{x_i+y}| \prod_{j=1}^{i-1} |\phi_{x_j+y}| \prod_{j=i+1}^{n} |\phi'_{x_j+y}|\,.
\]
Note that the absolute values on the right hand side cancel the effect of any possible complex conjugations on the left hand side.
Taking an expectation over $\gamma$ and then using Cauchy--Schwarz inequality 
and the natural order in $I$ to simplify the notations, we obtain
\begin{align*}
 & 
\left\langle \left| \phi^{I+y} - (\phi')^{I+y}  \right|\right\rangle_{\gamma}
 \le 
  \sum_{x\in I} \left\langle |\phi_{x+y}-\phi'_{x+y}|
  \prod_{x'<x} |\phi_{x'+y}| \prod_{x'>x} |\phi'_{x'+y}| \right\rangle_{\gamma}
\\
 &\quad \le   \sum_{x\in I} \left\langle |\phi_{x+y}-\phi'_{x+y}|^2\right\rangle_{\gamma}^{\frac{1}{2}}
\left\langle 
  \prod_{x'<x} |\phi_{x'+y}|^2 \prod_{x'>x} |\phi'_{x'+y}|^2 \right\rangle_{\gamma}^{\frac{1}{2}}
\\
 &\quad \le  \sum_{x\in I} \left\langle |\phi_{x+y}-\phi'_{x+y}|^2\right\rangle_{\gamma}^{\frac{1}{2}} 
  \prod_{x'<x} \left\langle|\phi_{x'+y}|^{q'}\right\rangle_{\gamma}^{\frac{1}{q'}}
  \prod_{x'>x} \left\langle|\phi'_{x'+y}|^{q'}\right\rangle_{\gamma}^{\frac{1}{q'}}
\,,
\end{align*}
where in the last step we have used the generalized H\"{o}lder's inequality with 
exponent $q'=2 (n-1)$ for which 
indeed $\sum_{x'\in I; x'\ne x} \frac{1}{q'}= \frac{1}{2}$ for all $x\in I$.

We may now conclude that the error $X$ is bounded by 
\[
 |X|\le \frac{1}{V} \sum_{y\in \Lambda}\sum_{x\in I}  
  \left\langle |\phi_{x+y}-\phi'_{x+y}|^2\right\rangle_{\gamma}^{\frac{1}{2}} 
  \prod_{x'<x} \left\langle|\phi_{x'+y}|^{q'}\right\rangle_{\gamma}^{\frac{1}{q'}}
  \prod_{x'>x} \left\langle|\phi'_{x'+y}|^{q'}\right\rangle_{\gamma}^{\frac{1}{q'}}\,.
\]
Here, only the first factor depends on $\gamma$, since all the other factors may be 
computed using the fixed marginal measures $\mu$ and $\mu'$.
Using the 
translation invariance of the marginal measures we obtain
\[
 |X|\le \frac{1}{V}\sum_{x\in I}  
  \prod_{x'<x} \left\langle|\phi_{x'}|^{q'}\right\rangle_{\mu}^{\frac{1}{q'}}
  \prod_{x'>x} \left\langle|\phi_{x'}|^{q'}\right\rangle_{\mu'}^{\frac{1}{q'}}
 \sum_{y\in \Lambda}
  \left\langle |\phi_{x+y}-\phi'_{x+y}|^2\right\rangle_{\gamma}^{\frac{1}{2}} \,.
\]

We next use the assumption that $A_n<\infty$ for the moments given in (\ref{eq:defAnmom}).
By translation invariance $A_n$ is independent of the choice of $x\in \Lambda$, and thus
by applying the Schwarz inequality to the sum over $y$, we obtain
\begin{align*}
& 
 |X|\le \frac{1}{V}A_n^{n-1}\sum_{x\in I}  \sqrt{V}
 \left[\sum_{y\in \Lambda}
 \mean{|\phi_{x+y}-\phi'_{x+y}|^2}_\gamma\right]^{\frac{1}{2}} 
\\ & \quad  
 =  \frac{1}{\sqrt{V}}A_n^{n-1} n
 \mean{\norm{\phi-\phi'}_2^2}_\gamma^{\frac{1}{2}} 
 =  L^{-\frac{d}{2}} A_n^{n-1} n
 \mean{\norm{\Phi-\Phi'}^2}_\gamma^{\frac{1}{2}} \,.
\end{align*}
Since the left hand side does not depend on the coupling $\gamma$, taking an infimum yields 
the bound in (\ref{eq:finmomestimate}), as before.  This concludes the proof of the Lemma.
\end{proof}

In the bound (\ref{eq:finmomestimate}) a factor $L^{\frac{d}{2}}$ gets cancelled from the 
Wasserstein distance.  Combined with the earlier results, the bound thus 
goes to zero if $n$ is not allowed to increase when taking $L\to \infty$, as long as the constants $A_n$ remain bounded in the limit.
As proven in Lemma \ref{th:momentbounds} at the end of the section, this holds for the measures considered here.
Hence, we may conclude that (\ref{eq:finmomestimate}) 
combined with the Wasserstein estimates stated in the main results in Sec.~\ref{sec:mainresults} 
implies that
\[
 \mean{\phi^I}_{\mu} = \mean{\phi^I}_{\mu'} + O(L^{-p'})\,,
\]
as $L\to \infty$  
if $\mu$ is a supercritical spherical model measure and $\mu'$ is a compatible factorized supercritical measure.
Summarizing all assumptions in one place, we obtain the following result as an immediate corollary of Corollary \ref{th:fixedomega}, Lemma \ref{th:localobsestimate},
and Lemma \ref{th:momentbounds}.
\begin{theorem}\label{th:infvollocalconv}
 Suppose that $d\ge 3$ and $\omega$ satisfies Assumption \ref{as:omega}, and consider any supercritical $\rho$ 
 as in Corollary \ref{th:fixedomega}.  Fix a maximum order $n\ge 1$ of the local moment.
 Then there are $C',p',L'>0$ for which the following holds:  if $L\ge L'$, 
 we may find a split $(\Lambda_0^*,\Lambda^*_+)$ of $\Lambda^*$ 
 and define the corresponding factorized supercritical measure $\mu_1$ by (\ref{eq:defmu1}) so that
\begin{align}\label{eq:finmomsupercrit}
 \left|\mean{\phi^I}_{\mu_0} - \mean{\phi^I}_{\mu_1}  \right| \le 
   C' L^{-p'}\, , 
\end{align}
 for any sequence $I$ from $\Lambda_L\times \{\pm 1\}$ of length at most $n$.
\end{theorem}
Using the constants occurring in Corollary \ref{th:fixedomega}, we may use 
$p'=\frac{d/2-1}{2 M_0+1}$
in (\ref{eq:finmomsupercrit}).  However, 
as discussed before the Corollary, this value might not always be optimal, i.e., the result could hold also for larger values of $p'$.

For applications of the approximation result, perhaps the most important consequence is the simplification of the structure of fluctuations.  
Namely, apart from the few condensate degrees of freedom, the field becomes Gaussian and translation invariant.  
In fact, as we will show next, 
its infinite volume statistics are given by the critical lattice field $\psi_x$, $x\in \Z^d$, which has zero mean and covariance with $\E[\psi_x \psi_y] = 0$ and 
\begin{align}\label{eq:defcritcov}
 \E[\psi_x \psi_y^*] = \int_{\T^d} \!\rmd k\, \frac{1}{\omega(k)-\ommin} \rme^{\ci 2\pi k\cdot (x-y)}\,,
\end{align}
for all $x,y\in \Z^d$.

More precisely, for all of the factorized supercritical measures in Sec.~\ref{sec:factsupercrit}, 
the field $\phi_x$ can be written as a sum of two independent random fields of which the normal fluid component $\phi^+$ is defined by 
$\phi^+_x = \int_{\Lambda^*_+}\!\rmd k\, \Phi^+_k \rme^{\ci 2\pi k \cdot x}$ where $\Phi^+$ is distributed according to the measure $\mu_+$ in (\ref{eq:defmuphiplus}).  
Therefore, for any compactly supported test function $J:\Z^d\to \C$,
we can define the random variable
\[
 \langle J,\phi^+ \rangle := \sum_{x\in \Z^d} J(x)^* \phi^+_x\,,
\]
as soon as $L$ is large enough so that $\Lambda_L$ contains the support of $J$.  Then $ \langle J,\phi^+ \rangle $ has mean zero  and a variance for which $\mean{\langle J,\phi^+ \rangle^2}=0$ and 
\[
 \mean{| \langle J,\phi^+ \rangle |^2} = \int_{\Lambda^*_+}\!\rmd k'\,\int_{\Lambda^*_+}\!\rmd k\, \E_{\mu_+}\!\!
 \left[\Phi_{k'}^* \Phi_k\right] \FT{J}(k') \FT{J}(k)^* = \int_{\Lambda^*_+}\!\rmd k\, \frac{1}{e_k}\left|\FT{J}(k)\right|^2\,,
\]
where
\[
 \FT{J}(k) := \sum_{z\in \Z^d} \rme^{-\ci 2\pi k\cdot x} J(x)\,.
\]
The function $\FT{J}:\T^d\to \C$ is continuous, hence also bounded.  
We assume that the split $(\Lambda_0^*,\Lambda^*_+)$ for all $L$ has the properties listed in Lemma \ref{th:mainomlemma}.
Then it is possible to partition $\T^d$ into
boxes of side length $\frac{1}{L}$ so that $\frac{1}{e_k}$ is bounded in the corresponding box by a constant times $\frac{1}{\omega-\ommin}$,
apart possibly from a finite number of boxes.  Due to the lower bound for $e_k$ valid for all $k\in \Lambda^*_+$,
we may ignore the exceptional boxes, and for the remaining ones use 
dominated convergence theorem to conclude that for any fixed $J$
\[
 \lim_{L\to\infty} \mean{| \langle J,\phi^+ \rangle |^2}  = \int_{\T^d} \!\rmd k\, \frac{1}{\omega(k)-\ommin}
 \left|\FT{J}(k)\right|^2\,.
\]
Details of this construction, as well as explicit estimates in $L$ for the size of the error, can be found in the proof of 
(\ref{eq:rhoerrorbound}) given at the end of Sec.~\ref{sec:proofofgap}.

Then an application of the polarization identity proves that for any two test functions $J_1$ and $J_2$ with a compact support we have
\[
 \lim_{L\to\infty} \mean{\langle J_1,\phi^+ \rangle^* \langle J_2,\phi^+ \rangle}  = \int_{\T^d} \!\rmd k\, \frac{1}{\omega(k)-\ommin}  \FT{J}_1(k) \FT{J}_2(k)^*\,.
\]
Restricted to single site test functions, we may thus conclude that (\ref{eq:defcritcov}) is indeed the limit of any pointwise covariances.  Since both the finite volume and the limit field are Gaussian, these results also immediately imply the convergence of all finite moments.

We conclude the section by showing that both the original and factorized fields have uniformly bounded moments.
\begin{lemma}\label{th:momentbounds}
  Suppose that $d\ge 3$ and $\omega$ satisfies Assumption \ref{as:omega}.
  Consider some supercritical $\rho$, some $L\ge L_0$ and any split $(\Lambda_0^*,\Lambda^*_+)$ of $\Lambda^*$ satisfying all properties stated in Lemma \ref{th:mainomlemma}.  Let $\mu$ be either $\mu_0$ or one of the measures $\mu_1$ or $\mu_1'$ defined for this split in Theorem \ref{th:mainresult} and Proposition \ref{th:mu1prop}.
  
  Then to each $m\ge 0$ there is an $L$-independent constant $c_m$ such that 
  \[
   \mean{|\phi_{x}|^{2 m}}_{\mu} \le c_m\,.
  \]
  for the random variable $\phi_x$ defined by (\ref{eq:defphix}) for any $x\in \Lambda_L$.
\end{lemma}
\begin{proof}
 If $m=0$, defining $c_0=1$ obviously suffices since $\mu$ is a probability measure.  Assume thus $m>0$.
 
 Split $\phi_x$ into a condensate and normal fluid component as follows
 \[
  \phi^0_x :=  \int_{\Lambda^*_0}\!\rmd k\, \Phi_k \rme^{\ci 2\pi k \cdot x}
  \qand
  \phi^+_x :=  \int_{\Lambda^*_+}\!\rmd k\, \Phi_k \rme^{\ci 2\pi k \cdot x}\,.  
 \]
Then $\phi_x = \phi^0_x+ \phi^+_x$, and the condensate component may be bound by using the upper bound $M_0$ from Lemma \ref{th:mainomlemma},
\[
 |\phi^0_x |\le \int_{\Lambda^*_0}\!\rmd k\, |\Phi_k|
 \le \sqrt{V_0/V} \norm{\Phi^0}
 \le \sqrt{M_0 \rho_0[\Phi]}\,.
\]
Under the measure $\mu_0$, $\rho_0[\Phi]\le \rho$ almost surely, and under either of the measures 
$\mu_1$ or $\mu_1'$ we have $\rho_0[\Phi]=\Delta\le \rho$ almost surely.
Therefore, in all of the three cases the condensate field is almost surely uniformly bounded in $L$,
$|\phi^0_x |\le\sqrt{M_0 \rho}$.

We then employ H\"{o}lder's inequality for the dual pair $(2m,2m/(2m-1))$ to bound the moment
\[
  \mean{|\phi_{x}|^{2 m}}_{\mu} \le  \mean{(|\phi^+_{x}|+|\phi^0_{x}|)^{2 m}}_{\mu}
  \le 2^{2m-1} \left(\mean{|\phi^+_{x}|^{2 m}}_{\mu}+\mean{|\phi^0_{x}|^{2 m}}_{\mu}\right)
  \,.
\]
The condensate term on the right hand side is now bounded by $(M_0 \rho)^{m}$, so it only remains to estimate the normal fluid term.

Let us begin with the case where $\mu$ is $\mu_1$ or $\mu_1'$.  Since $\phi^+_x$ only depends on $\Phi^+$,
the product structure of these two measures implies that 
\[
 \mean{|\phi^+_{x}|^{2 m}}_{\mu}=\mean{|\phi^+_{x}|^{2 m}}_{\mu_+}
 = \int_{(\Lambda_+^*)^m} \!\rmd k  \int_{(\Lambda_+^*)^m} \!\rmd k'
 \rme^{\ci 2\pi x\cdot\sum_{i=1}^m(k_i-k'_i)}\bigmean{\prod_{i=1}^m (\Phi_{k_i} \Phi^*_{k'_i})}_{\mu_+}\,.
\]
The remaining expectation is over independent, mean zero, Gaussian complex random variables.
By the Wick rule and gauge invariance, the expectation is zero unless there is a permutation $\pi$
of $\set{1,2,\ldots,m}$ such that $k'_i = k_{\pi(i)}$ for all $i$.  Therefore,
\[
 \bigmean{\prod_{i=1}^m (\Phi_{k_i} \Phi^*_{k'_i})}_{\mu_+} = \sum_{\pi\in S_m} \prod_{i=1}^m 
 \cf{k'_i=k_{\pi(i)}} \prod_{i=1}^m\frac{V}{e_{k_i}}\,.
\]
For any nonzero term in the sum, $\sum_{i=1}^m k'_i=\sum_{i=1}^m k_i$ and thus $ \rme^{\ci 2\pi x\cdot\sum_{i=1}^m(k_i-k'_i)}=1$, and summing over $k'$ yields
\[
 \mean{|\phi^+_{x}|^{2 m}}_{\mu}
 = \int_{(\Lambda_+^*)^m} \!\rmd k\sum_{\pi\in S_m} \prod_{i=1}^m\frac{1}{e_{k_i}} 
 = m! \rhoc(L)^m\le m! \rho^m\,.
\]
Therefore, for these two measures, we may use $c_m=2^{2m-1} (m!+M_0^m) \rho^m$.

It remains to consider the normal fluid contribution for $\mu=\mu_0$.  As above, we have
\[
 \mean{|\phi^+_{x}|^{2 m}}_{\mu}
 = \int_{(\Lambda_+^*)^m} \!\rmd k  \int_{(\Lambda_+^*)^m} \!\rmd k'
 \rme^{\ci 2\pi x\cdot\sum_{i=1}^m(k_i-k'_i)}\bigmean{\prod_{i=1}^m (\Phi_{k_i} \Phi^*_{k'_i})}_{\mu_0}\,,
\]
and by gauge invariance of $\mu_0$, the remaining expectation is zero unless for each $k\in \Lambda_+^*$ there are the same number of $\Phi_k$ and $\Phi_k^*$ terms in the product, in which case the product yields a positive number.  Thus for the nonzero terms also here we can find  a permutation $\pi$
of $\set{1,2,\ldots,m}$ such that $k'_i = k_{\pi(i)}$ for all $i$.
Therefore,
\[
0\le \bigmean{\prod_{i=1}^m (\Phi_{k_i} \Phi^*_{k'_i})}_{\mu_0} \le \sum_{\pi\in S_m} \prod_{i=1}^m 
 \cf{k'_i=k_{\pi(i)}} \bigmean{\prod_{i=1}^m |\Phi_{k_i}|^2}_{\mu_0}\,.
\]
Continuing as above, and observing that $\rho_+[\Phi] := \frac{1}{V}\int_{\Lambda_+^*} \!\rmd k |\Phi_k|^2\le N[\Phi]/V$ is bounded by $\rho$ almost surely under $\mu_0$, we find an upper bound
\[
 \mean{|\phi^+_{x}|^{2 m}}_{\mu}
 \le \int_{(\Lambda_+^*)^m} \!\rmd k\sum_{\pi\in S_m} V^{-m}
  \bigmean{\prod_{i=1}^m |\Phi_{k_i}|^2}_{\mu_0}
 = m! \mean{\rho_+[\Phi]^m}_{\mu_0}\le m! \rho^m\,.
\]
Therefore, also for $\mu=\mu_0$, we may use $c_m=2^{2m-1} (m!+M_0^m) \rho^m$.  Let us point out that by Lemma \ref{th:mainomlemma} $\rhoc(L)$ is bounded in $L$ and thus it is not a contradiction to assume that $\rho$ is fixed and supercritical for all $L\ge L_0$.
\end{proof}

\section{Example lattice dispersion relations}\label{sec:examples}

As an application, we consider explicitly a number of dispersion relations $\omega:\T^d\to \R$, all of which are continuous (periodic) functions.
Let us first recall that, once we define $\phi_x$ by (\ref{eq:defphix}), the energy and norm satisfy
\[
 H[\Phi] = \sum_{x,y\in \Lambda_L} \phi_x^* \alpha(x-y;L) \phi_y \qand 
 N[\Phi] = \sum_{x\in \Lambda_L} |\phi_x|^2\,,
\]
where 
\[
 \alpha(x;L) :=  \int_{\Lambda^*}\!\rmd k\,  \omega(k)
     \textrm{e}^{\textrm{i} 2\pi k \cdot x}\,.
\]
Taking $L\to \infty$ thus shows that $\alpha(x;L)\to \alpha(x)=\int_{\T^d}\!\rmd k\,  \omega(k) \rme^{\ci 2\pi k \cdot x}$ for each $x\in \Z^d$.  Here
$\alpha(x)$ are the Fourier coefficients of $\omega$ and they are $\ell_2$-summable since $\omega\in L^2(\T^d)$. In particular, $\alpha(x)\to 0$ as $|x|\to \infty$.  
Furthermore, if $\omega$ is a restriction of an analytic function, we may conclude that its Fourier coefficients $\alpha(x)$ are exponentially decreasing in $|x|\to \infty$, 
and all such functions correspond to ``short-range'' interactions for the field $\phi_x$.

\subsection{Nearest neighbour interactions}\label{sec:nninteractions}
 
In the original Berlin--Kac paper nearest neighbour interactions where considered which for a rectangular lattice 
would correspond to using a dispersion relation
\[
 \omega(k) = a + b \sum_{i=1}^d \sin^2(\pi k_i)\,,
\]
where $a\in \R$ and $b>0$. (Since $2 \sin^2 (\pi y)=1-\cos (2\pi y)=1-\frac{1}{2}(\rme^{\ci 2\pi y}+\rme^{-\ci 2\pi y})$, it is straightforward to check that then $|\alpha(x;L)|=0$ if $|x|_\infty>1$, i.e., for points which are not nearest neighbour on a rectangular lattice.)

Clearly, $\omega$ is twice continuously differentiable and $k=0$ is the unique minimum point on $\T^d$ and $\ommin=\omega(0)=a$.
Also, $D^2\omega(0)=2\pi^2 b \,1$ is proportional to the unit matrix $1$ and strictly positive.  Thus $\omega$ satisfies Assumption \ref{as:omega} with $T_0=\set{0}$.

For fixed $L\ge 2$, let us parameterize the dual lattice $\Lambda^*$ by $k=\frac{n}{L}$ where $n\in \Lambda_L$, in particular, $|n|_\infty\le \frac{L}{2}$.
Since $0\in \Lambda^*$, we have $\omega_0=\ommin=a$, and thus the excess energies satisfy
\[
 e_k = b \sum_{i=1}^d \sin^2\!\left(\frac{\pi n_i}{L}\right)\ge \frac{4 b}{L^2} \sum_{i=1}^d n_i^2\,.
\]
Therefore, defining $\Lambda_0^*=\set{0}$ and $\Lambda_+^*=\Lambda^*\setminus\set{0}$, results in a split of $\Lambda^*$ which is separated by 
the energy interval $[0,4 b L^{-2}]$ and thus has $\delta_L=0$.  We also have
\[
  \rhoc(L) = \int_{\Lambda^*_+}\!\rmd k\,  \frac{1}{e_k} \le \frac{L^{2-d}}{4 b} \sum_{1\le |n|_\infty\le \frac{L}{2}} |n|^{-2} = O(1)\,,
\]
and 
\[
 \frac{1}{V}\int_{\Lambda^*_+}\!\rmd k\,  \frac{1}{e_k^2} \le L^{4-2 d} (4b)^{-2} \sum_{1\le |n|_\infty\le \frac{L}{2}} |n|^{-4} \,.
\]
By a Riemann sum approximation (see Sec.~\ref{sec:proofofgap} for details) we find that the right hand side is $O(L^{-2})$, for $d=3$,
it is $O( L^{-4} \ln L)$ for $d=4$, and $O(L^{-d})$ for $d\ge 5$.  Hence, also $\vep_L$ satisfies these bounds, and we may apply Theorem \ref{th:mainresult}
for all large enough $L$.

We conclude that $W_2(\mu_0,\mu_1)\le C_2 L^{\frac{d}{2}-p'}$ with $p'=\frac{d}{4}$, for $d\ge 5$, any $p'<1$, for $d=4$, and $p'=\frac{1}{2}$, for $d=3$.
Since $V_0=1$ and $k=0$ is the unique condensate Fourier mode, we can then apply 
Proposition \ref{th:mu1prop} and Lemma \ref{th:localobsestimate}
to conclude that for any finite moment, i.e., for 
index sets $I$ whose length is less than some arbitrary cut-off, we can approximate
\[
\mean{\phi^I}_{\mu_0} = \mean{\psi^I} + O(L^{-p'})\,,
\]
where $\psi_x=\phi^+_x + \phi^0_x$ and $\phi^0_x = \sqrt{\rho -\rhoc(L)} \rme^{\ci \theta}$ is a constant field with a random phase.  As shown in Sec.~\ref{sec:Wassersteintomoments},
$\phi^+_x$ behaves like the critical Gaussian field.

\subsection{Acoustic phonon type interactions}

Although not covered by Assumption  \ref{as:omega}, we can also apply Theorem \ref{th:mainresult} directly by explicit estimates 
to the following dispersion relation which would appear in the theory of acoustic phonons:
\[
  \omega(k) = \left(\sum_{i=1}^d \sin^2(\pi k_i)\right)^{\frac{1}{2}}\,.
\]
By the computations in the previous subsection, then again $k=0$ is the unique minimum also on finite lattices and the excess energies satisfy
\[
 e_k \ge 2 L^{-1} |n|\,, \quad k = \frac{n}{L}\,,\ n\in \Lambda_L\,.
\]
Hence, for all $d\ge 2$, we have $\rhoc(L)=O(1)$ and $\vep_L = O(L^{-d})$ for $d\ge 3$ and $O( L^{-d} \ln L)$ for $d=2$.
Thus the approximation result given at the end of Sec.~\ref{sec:nninteractions} holds also in this case, only with smaller errors and including also the case $d=2$.

\subsection{Dispersion relation with several minima}\label{sec:dispwithseveralmin}

Let
\[
 \omega(k) = \sum_{i=1}^d \sin^2(2 \pi k_i)\,,
\]
which has $2^d$ global minima at points with $k_i\in \set{0,\frac{1}{2}}$ for $i=1,2,\ldots,d$.  All of these are non-degenerate and thus
$\omega$ satisfies  Assumption  \ref{as:omega}.  Also, $0$ is a minimum and thus for all $L$ the minimum value is reached,
$\ommin=0=\omega_0$.

Suppose first that $L$ is odd, say $L=2 m+1$ with $m\in \N_+$.  Then if $k_0\in T_0$ is not zero, it has some component $i$ such that $k_i=\frac{1}{2}$.
For such $i$ and any $n\in \Z^d$, we have $n_i - L k_i=n_i - m - \frac{1}{2}\ne 0$.  Hence, $T_0\cap \Lambda^* = \set{0}$.
In addition, if $1\le n_i\le m$,  we have
\[
  \sin\!\left(\frac{2 \pi n_i}{L}\right)  = 2   \sin\!\left(\frac{\pi n_i}{L}\right) \cos\!\left(\frac{\pi n_i}{L}\right)  
  \ge  \frac{\min(2 n_i,L-2 n_i)}{L} \ge \frac{1}{L}\,.
\]
Therefore, $e_k\ge L^{-2}$ for all $k\ne 0$, and one may modify the earlier estimates to prove that the split with
$\Lambda_0^*=\set{0}$ has $\vep_L=O(L^{-4 p'})$ with $p'$ chosen as for the nearest neighbour interactions.  Thus for odd $L$ one finds a single-component condensate, even though $|T_0|=2^d$.

If $L$ is even, say $L=2m$ with $m\in \N_+$, we have
$\frac{1}{2}=\frac{m}{L}$, and thus $T_0\subset \Lambda_L/L$.  Defining $\Lambda_0^*=T_0$ results in a split for which $\rhoc(L)=O(1)$ and 
$\vep_L=O(L^{-4 p'})$ as above but now the condensate is $2^d$-fold degenerate.  In addition, $e_k=0$ for each $k\in \Lambda_0^*$, so it is not possible to decrease $\Lambda_0^*$ without reducing the gap size to zero.  We can also apply item \ref{it:allground} of Proposition \ref{th:mu1prop}
and conclude that in the condensate the Fourier modes $k\in T_0$ are distributed uniformly on a sphere and hence the condensate field $\phi^0_x$
has strong oscillations in $x$.

In summary, the odd and even lattice sizes behave differently, and 
it does not really make sense to talk about $L\to\infty$ limit of the measure $\mu_0$, at least not without first removing the condensate modes.
This result becomes more transparent if one computes the coupling function $\alpha(x;L)$: these correspond to next-to-nearest neighbour couplings 
where $\alpha(x)=0$ unless $x=0$ or $|x|_\infty=2$.  Considering each of the $d$ directions separately, one observes that if $L$ is even, the odd and even sites become 
disconnected, and thus the system decouples into $2^d$ independent nearest neighbour systems.  On the other hand, if $L$ is odd, odd and even sites are coupled by 
``going around the circle once''.  In fact, this system corresponds to a single nearest neighbour lattice where the particle labels have been permuted.  
Since the estimates in Theorem \ref{th:mainresult} are sufficiently strong to distinguish between the two cases, we find that
they provide a reliable, relatively simple method of isolating the condensate modes also in this somewhat pathological setup.

\subsection{Dispersion relations with varying condensate energy}

As a straightforward generalization of the above dispersion relations, one can have any point $\zeta\in \T^d$ as the global minimum, for instance using
\[
 \omega(k;\zeta) = \sum_{i=1}^d \sin^2(\pi (k_i-\zeta_i))\,.
\]
Even though the minimum point is unique on the torus, if $\zeta\ne 0$, it does not need to belong to $\Lambda^*$ and then there might be several minimum points in $\Lambda^*$.

Consider for instance an odd $L=2m+1$ and $\zeta_i=\frac{1}{2}$ for all $i$.  Then $\frac{n_i}{L}-\frac{1}{2}= -\frac{1+2(m-n_i)}{2 L}$
and $\frac{n_i}{L}+\frac{1}{2}= \frac{1+2(m+n_i)}{2 L}$, and thus in this case $\omega_0=d \sin^2\!\frac{\pi}{2 L}$ and it is reached whenever 
$n_i=\pm m$ for all $i$.  Thus to the unique continuum minimum $\zeta$ there are $2^d$ minimum points in $\Lambda^*$.
In fact, in this case one should choose $\Lambda_0^*$ to consist of these $2^d$ points, since then for $k\in \Lambda^*_+$ the excess energies $e_k$ increase
like $|n|^2/L^2$ where $|n|$ denotes the number of ``lattice steps'' from $k$ to the set $\Lambda^*_0$, leading to similar estimates as in the nearest neighbour case.

If $L$ is even for this dispersion relation, $\zeta\in \Lambda^*$, $\omega_0=0$, $\Lambda_0^*=\set{\zeta}$, and the behaviour is identical to the nearest neighbour case.

Considering irrational minimum points $\zeta$ can lead to much more complicated situations.  
For example, suppose $r$ is an irrational number between $0$ and $\frac{1}{2}$ which has a binary representation $b_j\in \set{0,1}$, $j\in \N$, i.e., suppose that $r = \sum_{j=2}^\infty b_j 2^{-j}$ where the sequence $(b_j)$ does not converge to zero or one.
Set $\zeta_1=r$ and $\zeta_i=0$, for $i\ge 2$, and consider the following dispersion relation obtained as a product of two previous ones,
\[
 \omega(k) := \omega(k;0) \omega(k;\zeta)\,,
\]
with global minima at $0$ and $\zeta$.
Then for each $L$, $0\in \Lambda^*$, $\omega_0=0$, and this value can only be reached at $k=0$ on $\Lambda^*$.  However, for values of $k=n/L$, with $n_i=0$ for $i\ge 2$, we have
\[
 \omega(k) \le \sin^2 \frac{\pi (n_1-L r)}{L}\,.
\]
Along the subsequence $L=2^N$, $N\in \N$, here $n_1-L r = n_1-\sum_{\ell=0}^{N-2} b_{N-\ell} 2^\ell -\sum_{j=1}^{\infty} b_{N+j} 2^{-j}$.
We can choose $n_1=\sum_{\ell=0}^{N-2} b_{N-\ell} 2^\ell\le \sum_{\ell=0}^{N-2} 2^\ell = 2^{N-1}-1<\frac{L}{2}$, and for this value
\[
 e_k = \omega(k) \le \pi^2 \left(\sum_{j=1}^\infty b_{N+j} 2^{-j-N}\right)^2\,.
\]
Hence, by considering a binary sequence with ever less frequent ones and sufficiently large $N$, the bound can be made proportional to $L^{-p}$ for any $p\ge 2$.  Depending on how small the term is, the above point $k=n/L$ might or might not belong to the condensate modes $\Lambda_0^*$.
In particular, there are instances for which $e_k>0$ but $e_k \le \frac{1}{2\rho} L^{-2 d}$, and thus item \ref{it:smalldev} of Proposition \ref{th:mu1prop} can be applied without increasing the magnitude of the error. Hence, the system behaves like a uniformly distributed two-component condensate
even though $e_k$, $k\in \Lambda^*_0$, is not identically zero.

\subsection{Anisotropic dispersion relations}

Another generalization of the above condensate cases is to consider anisotropic dispersion relations.
For instance, in addition to shifting the global minimum to $\zeta\in \T^d$ we may 
take any finite collection of points $M^{(\ell)}\in \Z^d$, $\ell=1,2,\ldots,N$, choose some 
weights $b_\ell>0$ for them, and define
\[
 \omega(k) = 
 \sum_{\ell=1}^N b_\ell \sin^2(\pi (k-\zeta)\cdot M^{(\ell)})\,.
\]
If there is a sufficient variety of points in the collection, for instance, if all unit vectors are included,
there is only one global minimum for this dispersion relation, located at $k=\zeta$.
The Hessian at this point is equal to 
\[
 2\pi^2 \sum_{\ell=1}^N b_\ell M^{(\ell)} \otimes M^{(\ell)}\,,
\]
so that the second derivative into a direction $v\in S^{d-1}$ at $k=\zeta$ is given by
\[
 2\pi^2 \sum_{\ell=1}^N b_\ell \left|v\cdot M^{(\ell)}\right|^2\,.
\]
Hence, essentially arbitrary asymmetries between different directions 
may be generated near the minimum point by varying $m$ and $b$.

In the proof of Lemma \ref{th:mainomlemma} given in Sec.~\ref{sec:proofofgap}, the uniform upper bound for the  
number of degrees included in the condensate, $M_0$, depends on the dimension but also on 
the ratio between the maximal and minimal eigenvalue of the Hessian of $\omega$ at its minima, i.e., 
on the maximal anisotropy at these points.
The value appearing in the proof typically overestimates the true number of degrees of freedom needed. Let us conclude with two examples
which highlight the problems which arise when trying to improve on such general uniform bounds.

For simplicity, let us consider anisotropy in the first two components only. 
To borrow results from the previous computations, assume that $L=2m+1$ is odd and take $\zeta=(\frac{1}{2},0,0,\ldots,0)$.
We reparameterize the first component using $m_1:= m-L |k_1|\in \Z$ and the sign $\sigma_1$ of $k_1$.
Then $0\le m_1\le m$ and $k_1=\sigma_1 |k_1| = \sigma_1 ( m-m_1)/L$, implying also that
$|\sin(\pi (k_1-\frac{1}{2}))|= |\sin(\pi (2 m_1+1)/(2L))|\ge (2 m_1+1)/L$.

We first consider the nearest neighbour case where the first component has unit weight but the rest have a much smaller weight $1/B$, where
$B\gg 1$.
Then for $k\in \Lambda_L$, and denoting $n_i = L k_i$, for $i=2,3,\ldots,d$, we find an approximation
\[
\omega(k) \approx \pi^2 \frac{(m_1+1/2)^2+n^2/B}{L^2}\,,
\]
valid for $m_1/L, |n|/L\ll 1$.
Thus the minimum value is reached at the two points where $m_1=0$ and $n=0$.
However, if $m_1=0$, we then also have $e_k\approx \pi^2\frac{n^2}{B L^2}$ whenever $|n|/L\ll 1$.
Suppose that we wish to include in the condensate $\Lambda_0^*$ at least all $k$ with $e_k\le L^{\frac{1}{2}-d}$ 
(corresponding roughly to the choice $\kappa=\frac{1}{2}$ in Lemma \ref{th:mainomlemma}).
Since for some finite $L$ it can happen that $B\ge L^{d-1}$, the number of condensate modes can temporarily be very large.
This effect can be traced back to the flatness of constant level surfaces of $\omega$ caused by the strong anisotropy.

In the second example, we take also the first direction to have a small weight $1/B$ but add one more point to the collection:
set $b_{d+1}=1$ and $M^{(d+1)}=(M_1,M_2,0,\ldots,0)$ where $M_1,M_2\in \N$ are such that $M_2$ is odd and $M_1$ is even.
Suppose also that $L$ is large enough, satisfying $L\gg M_1,M_2$. Then for $m_1/L, |n|/L\ll 1$ we have
\[
 \omega(k) \approx \pi^2 \frac{(m_1+1/2)^2+n^2}{B L^2} + \pi^2 \frac{K^2}{L^2} \,,
\]
where, using the assumption that $M_1$ is an integer,
\[
 K = \left(L \left[k_1-\frac{1}{2}\right] M_1 + L k_2 M_2\right) \bmod L =
 \left(-\sigma_1 \left[m_1+\frac{1}{2}\right] M_1 + n_2 M_2\right) \bmod L\,.
\]
Since $M_1$ is even, $M_2$ is odd, and both are positive, we may 
set $n_2=\sigma_1\frac{M_1}{2}$ and choose $m_1$ so that $2 m_1+1=M_2$.  Setting also $n_i=0$ for $i\ge 3$, 
we obtain two points in $\Lambda_L^*$ for which 
$K=0$ and 
\[
 \omega(k) \approx \pi^2 \frac{M_1^2+M_2^2}{4 B L^2}\,.
\]
However, for any point for which $K\ne 0$, for instance, if $m_1=0=n_2$, we have
\[
 \omega(k)\ge \frac{4}{L^2}\,.
\]
Therefore, if the system is sufficiently anisotropic, e.g., $B\ge M_1^2+M_2^2$, it can happen that 
the minimum point is \defem{not} the nearest lattice point to the minimum on $\T^d$, but it could be found many lattice steps away from it.
In contrast to the first example, this effect does not disappear when $L\to \infty$, but will persists for all sufficiently large odd $L$
in the present case.

\section{Proof of the main result, Theorem \ref{th:mainresult}}\label{sec:proofofmain}

\begin{proofof}{Theorem \ref{th:mainresult}:}
Consider a fixed $L$ and a split $(\Lambda_0^*,\Lambda^*_+)$ of $\Lambda^*$ which is
separated by the energy interval $[a,b]$, $0\le a<b$, and has a relative energy gap $\delta^{-1}$.
We aim at separation in the degrees of freedom 
related to these two sets.  

We begin by simplifying the representation of the Berlin--Kac measure $\mu_0$.  Starting from the simplified form, we then construct a change of variables which will bring it closer to the measure $\mu_1$. We first shift the position of the $\delta$-constraint to match that in $\mu_1$.  This will introduce a shift in the normal fluid energies which we will need to repair back to the critical ones by a second change of variables.
Even after these changes, the measures will differ by a weight function which, however,
is close to one with high probability.  This property is checked quantitatively in a technical Lemma \ref{th:gmainlemma}, resulting in the estimates in Corollary \ref{th:mainthgbounds}.
To make the final comparison, we use the change of variables to construct a coupling 
between $\mu_0$ and $\mu_1$ which, together with Corollary \ref{th:mainthgbounds}, 
will result in the stated bound on their Wasserstein distance.

To begin, let us collect the field values for $k\in \Lambda_+^*$ into a vector $\Phi^+$, corresponding to the normal fluid,
and those for $k\in \Lambda_0^*$ into a vector $\Phi^0$, corresponding to the condensate.
We denote
\[
 V_0 := |\Lambda_0^*| \,, \qquad V_+ := |\Lambda_+^*| \,,
\]
for which $V_0,V_+>0$, and $V=V_0+V_+$.  Define also
\[
 N_0[\Phi] := \int_{\Lambda_0^*}\!\rmd k\, |\Phi_k|^2\,,\quad 
 N_+[\Phi] := \int_{\Lambda_+^*}\!\rmd k\, |\Phi_k|^2
 \,,\quad 
 \rho_+[\Phi] := \frac{N_+[\Phi]}{V}
  \,,\quad 
 \rho_0[\Phi] := \frac{N_0[\Phi]}{V}
 \,.
\]
Since $N_+[\Phi]+N_0[\Phi] = N[\Phi]$, we have now
\[
 H[\Phi] = \int_{\Lambda^*}\!\rmd k\, \omega(k) |\Phi_k|^2
 = \int_{\Lambda^*}\!\rmd k\, e_k |\Phi_k|^2 +
 \omega_0 N[\Phi]\,.
\]
Denote
\[
 E_+[\Phi] := \int_{\Lambda_+^*}\!\rmd k\, e_k |\Phi_k|^2\,,\qquad 
 E_0[\Phi] := \int_{\Lambda_0^*}\!\rmd k\, e_k |\Phi_k|^2
 \,,
\]
and we may conclude that in the integrand, in which almost surely $N[\Phi] = \rho V$,
we have
\[
 H[\Phi] = E_+[\Phi] + E_0[\Phi] + \omega_0 \rho V\,.
\]
Therefore, we may rewrite
\[
 \mu_0[\rmd \Phi] = \frac{1}{Z_{0}} 
 \prod_{k\in \Lambda_+^*} \left[\rmd \Phi^*_k \rmd \Phi_k\right]  \rme^{-E_+[\Phi]}
 \prod_{k\in \Lambda_0^*} \left[\rmd \Phi^*_k \rmd \Phi_k\right]  \rme^{-E_0[\Phi]} \delta(\rho_0[\Phi] + \rho_+[\Phi] - \rho)\,,
\]
where the new normalization constant is related to 
the one given in (\ref{eq:defmu0}) by $Z_0 = V \rme^{\omega_0 \rho V} Z_\rho $.

Let $\rhoc>0$ denote the critical density, measured as an expectation of $\rho_+$ over the probability measure (\ref{eq:defmuphiplus}), i.e., over
\[
 \mu_+[\rmd \Phi] := \frac{1}{Z_{+}} 
 \prod_{k\in \Lambda_+^*} \left[\rmd \Phi^*_k \rmd \Phi_k\right]  \rme^{-E_+[\Phi]}\,.
\]
By assumption, $e_k \ge b>0$ for each $k\in \Lambda_+^*$, and thus this is a well-defined Gaussian measure
under which $\re \Phi_k$, $\im \Phi_k$, $k\in\Lambda_+^*$, form a collection of jointly independent random variables,
with a zero mean and a variance $\frac{V}{2 e_k}$.  Therefore,
\[
 \mean{\rho_+}_{\mu_+} = 
 \frac{1}{Z_{+}} 
 \int \prod_{k\in \Lambda_+^*} \left[\rmd \Phi^*_k \rmd \Phi_k\right]  \rme^{-E_+[\Phi]} \frac{1}{V^2} \sum_{k\in \Lambda^*_+} |\Phi_k|^2
 = \frac{1}{V^2} \sum_{k\in \Lambda^*_+} \frac{V}{e_k} =  \int_{\Lambda_+^*}\!\!\rmd k\, \frac{1}{e_k}=\rhoc(L)\,,
\]
as defined in (\ref{eq:defrhocrit}).

Set then $\Delta:=\rho-\rhoc$, which is strictly positive by assumption.
Then we define the target measure $\mu_1$ as a product between $\mu_+$ and a suitably chosen condensate measure: 
we set
\begin{align}\label{eq:defmu1inproof}
 & \mu_1[\rmd \Phi] := \frac{1}{Z_1} 
 \prod_{k\in \Lambda_+^*} \left[\rmd \Phi^*_k \rmd \Phi_k\right]  \rme^{-E_+[\Phi]} 
 \nonumber \\ & \qquad \times
 \prod_{k\in \Lambda_0^*} \left[\rmd \Phi^*_k \rmd \Phi_k\right] \rme^{-E_0[\Phi]-\tilde{\vep}[\Phi] V \rhoc}
 \prod_{k\in \Lambda^*_+} \left(1-\frac{\tilde{\vep}[\Phi]}{e_{k}} \right)^{-1}
  \delta(\rho_0[\Phi] - \Delta) \,,
\end{align}
where $\tilde{\vep}$ depends only on the condensate components $\Phi^0$,
\[
  \tilde{\vep}[\Phi] := \frac{E_0[\Phi]}{N_0[\Phi]} \le \max_{k\in \Lambda^*_0}e_k \le a\,.
\]
Thus, for any $k\in \Lambda^*_+$,
\[
 \frac{\tilde{\vep}[\Phi]}{e_{k}} \le \delta <1\,,
\]
which implies that the weight in (\ref{eq:defmu1inproof}) is a strictly positive function.
Since here $\tilde{\vep}[\Phi]=E_0[\Phi]/(V \Delta)$ almost surely, this measure indeed coincides with the definition given in (\ref{eq:defmu1}).

To construct a suitable coupling between the measures $\mu_0$ and $\mu_1$, we rely on a change of variables and the diagonal concentration trick which we learned from 
Saksman and Webb, from the proof of Lemma B.1 in \cite[Appendix B]{SaksmanWebb2016}.
The trick is to construct an explicit coupling between two probability measures by concentrating as much of their common mass as possible in the diagonal of the coupling ($\Phi'=\Phi$) and distributing any remaining mass as a product on the off-diagonal ($\Phi'\ne \Phi$).
Although this coupling is seldom optimal, it can provide a good estimate
of the Wasserstein distance of the two measures in case most of the mass can be 
concentrated in the diagonal: note that the diagonal mass does not contribute to the value of the integral defining the Wasserstein distance in   (\ref{eq:defWpdist}).

In our application of the trick, we first need to change into variables using which 
the two measures share enough common mass.
To find new variables better adapted to compare the measures $\mu_0$ and $\mu_1$,
let us start from the measure $\mu_0$ and denote its integration variable by $\Psi$. 
The goal is to find a change of variables $\Psi=G[\Phi]$ which would yield a measure
close to $\mu_1$: we try to construct $G$ so that for any observable $f$ we would have
$\int\! \mu_0[\rmd \Psi] \, f(\Psi) = \int\! \mu_1[\rmd \Phi]\,  g[\Phi] f(G[\Phi])$ for some function $g$
which is close to one with high $\mu_1$-probability.  Some preliminary estimates and definitions will be needed to find the right choice, and we postpone the precise construction of the coupling later, until  Eq.~(\ref{eq:defmaingamma}).  

First, let us recall that $\Delta= \rho-\rhoc> 0$ and define 
\[
 \alpha[\Psi] := \begin{cases}      
   \frac{\rho_+[\Psi]-\rhoc}{\rho-\rhoc}\,, & \text{if }\rho_+[\Psi] < \rho\,, \\
      0\,, & \text{if }\rho_+[\Psi] \ge \rho\,.
 \end{cases}
\]
Note that $\alpha[\Psi]$ depends only on $\Psi^+$, and $-\frac{\rhoc}{\Delta}\le \alpha[\Psi]<1$.
Consider the expectation of some continuous function $f(\Psi^+,\Psi^0)$ with a compact support under the original measure $\mu_0[\rmd \Psi]$.
The mass constraint function can be written as
\[
 \rho_0[\Psi] + \rho_+[\Psi] - \rho
 = \rho_0[\Psi] -  (1-\alpha[\Psi]) \Delta\,,
\]
whenever $\rho_+[\Psi] < \rho$.  On the other hand, the set $\defset{\Psi^+}{\rho_+[\Psi]=\rho}$ has a measure zero, and if $\rho_+[\Psi] > \rho$, the mass constraint cannot be satisfied for any $\Psi^0$.  Hence, the collection of $\Psi$ with $\rho_+[\Psi]\ge \rho$ has zero measure with respect to $\mu_0$.
Since $\alpha[\Psi]<1$ depends only on $\Psi^+$, it is straightforward to make a change of variables $\Psi_k = \sqrt{1-\alpha[\Psi]} \Phi_k$ for $k\in \Lambda^*_0$.
Then $\rho_0[\Psi] = (1-\alpha[\Psi]) \rho_0[\Phi]$ and
\[
 \delta( \rho_0[\Psi] + \rho_+ - \rho) = \delta((1-\alpha[\Psi])(\rho_0[\Phi]- \Delta ))
 = \frac{1}{1-\alpha[\Psi]} \delta(\rho_0[\Phi]- \Delta )\,.
\]
More detailed discussion about the validity of this formula can be found in Appendix \ref{sec:defdelta}.  In particular, we are allowed to apply the formal rule for $\delta$-functions to take out the factor $(1-\alpha[\Psi])$ here since the $\delta$-function
can be integrated out using $\Psi^0$ while keeping $\Psi^+$, and hence also $\alpha[\Psi]$, fixed.

In the above change of variables, $E_0[\Psi] = (1-\alpha[\Psi])E_0[\Phi]$, and therefore we obtain 
\begin{align*}
 & \mean{f}_{\mu_0} = \frac{1}{Z_{0}}  \int 
 \prod_{k\in \Lambda_+^*} \left[\rmd \Psi^*_k \rmd \Psi_k\right]  \rme^{-E_+[\Psi]} \cf{\rho_+[\Psi] < \rho}
 \\ & \quad \times
 (1-\alpha[\Psi])^{V_0-1}\int
 \prod_{k\in \Lambda_0^*} \left[\rmd \Phi^*_k \rmd \Phi_k\right]  \rme^{-(1-\alpha[\Psi])E_0[\Phi]}
  \delta(\rho_0[\Phi]- \Delta ) f(\Psi^+,\sqrt{1-\alpha[\Psi]} \Phi^0)\,.
\end{align*}
We then use Fubini's theorem to change the order of $\Psi$ and $\Phi$ integrals.  Then we can simplify the integral by making a change of variables
for $\Psi^+$ using a fixed $E_0=E_0[\Phi]$ and assuming $\rho_0=\Delta$.  In particular, for $\rho_+[\Psi] < \rho$, we have 
\[
 E_0 \alpha[\Psi] = 
 \frac{E_0}{\Delta}(\rho_+[\Psi]-\rhoc) =  \frac{E_0}{\rho_0}(\rho_+[\Psi]-\rhoc) = \tilde{\vep} V (\rho_+[\Psi]-\rhoc)
 =  \tilde{\vep} N_+[\Psi]-\tilde{\vep} V \rhoc\,.
\]
Therefore,
\[
 \rme^{-E_+[\Psi]+  E_0 \alpha[\Psi]}  = \rme^{-\tilde{\vep} V \rhoc} \exp\biggl(-\frac{1}{V}\sum_{k\in \Lambda^*_+} (e_k-\tilde{\vep})|\Psi_k|^2\biggr)\,.
\]

We now make a second change of variables to correct for the shift of energies here: $\Phi_k = \sqrt{1-\tilde{\vep}/e_k}\Psi_k$ for $k\in \Lambda^*_+$.  As pointed out above, here $\tilde{\vep}/e_k<1$ and we can resolve the change of variables as easily as in the first case.  We find that 
\begin{align*}
 & \mean{f}_{\mu_0} = \frac{1}{Z_{0}} 
  \int \prod_{k\in \Lambda_0^*} \left[\rmd \Phi^*_k \rmd \Phi_k\right]  \rme^{-E_0[\Phi]}
  \delta(\rho_0[\Phi]- \Delta ) \rme^{-\tilde{\vep} V \rhoc} \prod_{k\in \Lambda^*_+} \left(1-\frac{\tilde{\vep}}{e_{k}} \right)^{-1}
 \\ & \quad \times  
  \int 
 \prod_{k\in \Lambda_+^*} \left[\rmd \Phi^*_k \rmd \Phi_k\right]  \rme^{-E_+[\Phi]} \cf{\rho_+ < \rho}
 (1-\alpha)^{V_0-1}
 f((1-\tilde{\vep}/e_k)^{-1/2}\Phi^+_k,\sqrt{1-\alpha} \Phi^0)\,,
\end{align*}
where $\tilde{\vep}=\tilde{\vep}[\Phi]$, and we need to substitute in the integrand
\[
 \text{``}\rho_+\text{''} = \frac{1}{V^2} \sum_{k\in \Lambda^*_+} \frac{1}{1-\frac{\tilde{\vep}}{e_{k}} } |\Phi_k|^2\,,\quad
 \text{``}\alpha\text{''} =  \frac{\rho_+-\rhoc}{\rho-\rhoc}\,,
\]
which are functions of both $\Phi^+$ and $\Phi^0$.

To summarize the result, let us define the functions
\[
 \rho'[\Phi] := \frac{1}{V^2} \sum_{k\in \Lambda^*_+} \frac{e_k}{e_k- \tilde{\vep}[\Phi]} |\Phi_k|^2\,,
 \qquad \alpha'[\Phi] := \frac{\rho'[\Phi]-\rhoc}{\rho-\rhoc}\,,
\]
and, using these, the weight function
\begin{align}\label{eq:defgweight}
 g[\Phi] := \frac{Z_1}{Z_{0}}
 \cf{\rho'[\Phi] < \rho}  (1-\alpha'[\Phi])^{V_0-1}
\end{align}
and the change of variables
\begin{align}\label{eq:defGchange}
 G(\Phi)_k :=
 \begin{cases}
             \left(1-\frac{\tilde{\vep}[\Phi]}{e_{k}}\right)^{-\frac{1}{2}} \Phi_k\,,& \text{for }k\in \Lambda^*_+\,,\\
             \left(1-\alpha'[\Phi]\right)^{\frac{1}{2}} \Phi_k\,,& \text{for }k\in \Lambda^*_0\,.            
             \end{cases} 
\end{align}
Then the above computation shows that
\begin{align}\label{eq:mu0frommu1}
& \mean{f}_{\mu_0} 
 = \int\mu_1[ \rmd \Phi] \, g[\Phi] f(G[\Phi]) \,.
\end{align}
Since $0\le  g[\Phi] \le \frac{Z_1}{Z_{0}} \left(\frac{\rho}{\Delta}\right)^{V_0-1}$, we can then use dominated convergence theorem to conclude that in fact (\ref{eq:mu0frommu1}) holds for all bounded continuous functions $f$.

Note that due to the change of variables implied by $G$ there is a shift in the position of the $\delta$-weight.
Therefore, the formula does not imply that $\mu_0$ or $\mu_1$ would be absolutely continuous with respect to each other (in fact, they are not: the collection of $\Phi$ with $\rho_+[\Phi]>\rho$ has zero measure with respect to $\mu_0$ but its measure is non-zero with respect to $\mu_1$; conversely, the collection of $\Phi$ with $\rho_0[\Phi]\le\frac{\Delta}{2}$ has zero measure with respect to $\mu_1$ but non-zero measure with respect to $\mu_0$).
However, as we will prove next in Lemma \ref{th:mainthgbounds}, the weight $g$ is close to one with high $\mu_1$-probability, and although there can be regions where it deviates significantly from one, $g$ remains always uniformly bounded.  These estimates will provide sufficient control for using the diagonal coupling trick at the end of the section, in (\ref{eq:defmaingamma}).

\begin{lemma}\label{th:gmainlemma}
Using the above definitions, we have
\begin{align}
  & - \frac{V_0-1}{1-\delta}    \left(\frac{\rho}{\Delta}\right)^{V_0-1}\sqrt{\tilde{\delta}}
 \le
 1-\frac{Z_0}{Z_{1}} \le  \frac{V_0}{1-\delta} \frac{\rho}{\Delta} \sqrt{\tilde{\delta}}\,,
 \label{eq:ZZbound}  \\ &
 \mean{|1-g|^2}_{\mu_1} \le \frac{1}{(1-\delta)^2}
 \left[\left(\frac{\rho}{\Delta}\right)^{2} + 4
 V_0^2 \left(\frac{\rho}{\Delta}\right)^{2 V_0}
  \left( \frac{Z_1}{Z_{0}}\right)^2\right]\tilde{\delta}\,,  \label{eq:gnearonebound} \\
& \mean{(\alpha')^2}_{\mu_1} \le \frac{\rho^2}{\Delta^2 (1-\delta)^2} 
   \tilde{\delta}
  \,,
\end{align}
 where
\begin{align}\label{eq:deftildedelta}
  \tilde{\delta} := 2  \delta +
 \frac{1}{V^2\rhoc^2} \sum_{k\in \Lambda^*_+} \frac{1}{e_k^2}\,.
\end{align}
\end{lemma}
\begin{proof}
Using $f=1$ in (\ref{eq:mu0frommu1}), we find that $\mean{g}_{\mu_1}= 1$, and thus
 \[
  \frac{Z_0}{Z_{1}} =  \mean{\cf{\rho' < \rho} (1-\alpha')^{V_0-1}}_{\mu_1}\,,
 \]
where  $-\frac{\rho_c}{\Delta}\le \alpha'<1$, and hence $0<1-\alpha'\le 1+\frac{\rho_c}{\Delta} = \frac{\rho}{\Delta}$.
Therefore,
 \[
 1-\frac{Z_0}{Z_{1}} = 
 \mean{\cf{\rho' \ge \rho} }_{\mu_1}
 + \mean{\cf{\rho' < \rho} \left[1-(1-\alpha')^{V_0-1}\right]}_{\mu_1}\,,
 \]
which implies that
\begin{align*}
 & - \mean{\cf{\rho' < \rho,\, \alpha' < 0} \left[(1-\alpha')^{V_0-1}-1\right]}_{\mu_1}
 \le
 1-\frac{Z_0}{Z_{1}}
 \\ & \quad 
 \le
 \mean{\cf{\rho' \ge \rho} }_{\mu_1}
 + \mean{\cf{\rho' < \rho,\, \alpha' > 0} \left[1-(1-\alpha')^{V_0-1}\right]}_{\mu_1}\,. 
\end{align*}
On the left hand side, the integrand is zero unless $-\frac{\rho_c}{\Delta}\le\alpha'< 0$. Thus
either $V_0=1$ and the term is always zero, or we may bound in the integrand
$(1-\alpha')^{V_0-1}-1\le |\alpha'| (V_0-1)  (\frac{\rho}{\Delta})^{V_0-2}$. Thus the 
expectation is bounded from above by $ (V_0-1)  (\frac{\rho}{\Delta})^{V_0-2} \mean{|\alpha'| }_{\mu_1}$.
On the right hand side, for $\alpha'>0$ we have
$0\le 1-(1-\alpha')^{V_0-1}\le |\alpha'| (V_0-1)$, and for $\rho'\ge \rho$, it holds that $\alpha'\ge 1$.
Therefore,
\[
 \mean{\cf{\rho' \ge \rho} }_{\mu_1}
 + \mean{\cf{\rho' < \rho,\, \alpha' > 0} \left[1-(1-\alpha')^{V_0-1}\right]}_{\mu_1}
 \le V_0 \mean{|\alpha'|}_{\mu_1} \,.
 \]

We have obtained the bounds
\[
  -  (V_0-1)  \left(\frac{\rho}{\Delta}\right)^{V_0-2} \mean{|\alpha'| }_{\mu_1}
 \le
 1-\frac{Z_0}{Z_{1}} \le V_0 \mean{|\alpha'|}_{\mu_1} \,,
\]
which imply also that 
\[
 \left|1-\frac{Z_0}{Z_{1}}\right|^2 \le \max\left(V_0^2,
 (V_0-1)^2  \left(\frac{\rho}{\Delta}\right)^{2 V_0-4}\right) \mean{|\alpha'| }^2_{\mu_1}
 \le V_0^2 \left(\frac{\rho}{\Delta}\right)^{2 (V_0-2)_+}\mean{|\alpha'|^2}_{\mu_1} \,,
\]
where $(r)_+:= r \cf{r>0}$. 
We may use this result and similar techniques to derive an upper bound for
 \begin{align*}
& \mean{|1-g|^2}_{\mu_1} =  \mean{\cf{\rho' \ge \rho} }_{\mu_1}
 + \mean{\cf{\rho' < \rho} |1-(1-\alpha')^{V_0-1} Z_1/Z_0|^2}_{\mu_1}
 \\ & \quad
 \le  \mean{|\alpha'|^2\cf{\rho' \ge \rho} }_{\mu_1}
 +2 \left( \frac{Z_1}{Z_{0}}\right)^2 \left( \left| \frac{Z_0}{Z_{1}}-1\right|^2
 + \mean{\cf{\rho' < \rho} |1-(1-\alpha')^{V_0-1}|^2}_{\mu_1}\right)
 \\ & \quad 
 \le \mean{|\alpha'|^2}_{\mu_1} \left[1 + 4
 V_0^2 \left(\frac{\rho}{\Delta}\right)^{2(V_0-2)_+}
  \left( \frac{Z_1}{Z_{0}}\right)^2\right]
 \,.  
\end{align*}

It remains to estimate
\[
 \Delta^2 \mean{(\alpha')^2}_{\mu_1} =  \mean{(\rho'-\rhoc)^2}_{\mu_1}\,,
\]
where
 \begin{align*}
& \rho' -\rhoc = \frac{1}{V^2} \sum_{k\in \Lambda^*_+}
 \frac{e_k}{e_k- \tilde{\vep}[\Phi]} |\Phi_k|^2 - \frac{1}{V} \sum_{k\in \Lambda^*_+} 
 \frac{1}{e_k}\,.
 \end{align*}
Since $\frac{\tilde{\vep}[\Phi]}{e_{k}} \le \delta$, here
 \begin{align*}
& \mean{(\rho'-\rhoc)^2}_{\mu_1}
  = \frac{1}{V^4} \sum_{k,k'\in \Lambda^*_+}
 \bigmean{\frac{1}{1- \tilde{\vep}/e_k} \frac{1}{1- \tilde{\vep}/e_{k'}} |\Phi_k|^2 |\Phi_{k'}|^2}
 \\ & \qquad
 - 2 \frac{1}{V^3} \sum_{k,k'\in \Lambda^*_+} \frac{1}{e_{k'}}
 \bigmean{\frac{1}{1- \tilde{\vep}/e_k}|\Phi_k|^2}
 + \frac{1}{V^2}\sum_{k,k'\in \Lambda^*_+} \frac{1}{e_{k'} e_k}
 \\ & \quad \le
 \frac{1}{(1-\delta)^2}\frac{1}{V^4} \sum_{k,k'\in \Lambda^*_+}
 \bigmean{|\Phi_k|^2 |\Phi_{k'}|^2}
 - 2 \frac{1}{V^3} \sum_{k,k'\in \Lambda^*_+} \frac{1}{e_{k'}}
 \bigmean{|\Phi_k|^2}
 + \frac{1}{V^2}\sum_{k,k'\in \Lambda^*_+} \frac{1}{e_{k'} e_k}\,.
 \end{align*}
The remaining Gaussian expectations can be computed explicitly, yielding for $k\ne k'$
\begin{align}\label{eq:phievenexp}
 \bigmean{|\Phi_k|^2} = \frac{V}{e_k}\,,\quad
 \bigmean{|\Phi_k|^2 |\Phi_{k'}|^2} = \frac{V^2}{ e_k e_{k'}}\,,\quad
 \bigmean{|\Phi_k|^4} = 2\frac{V^2}{e_k^2}\,. 
\end{align}
Therefore,
 \begin{align*}
& \Delta^2 \mean{(\alpha')^2}_{\mu_1} \le
 \frac{1}{(1-\delta)^2}\frac{1}{V^2} 
 \sum_{k,k'\in \Lambda^*_+} \frac{1}{ e_k e_{k'}}
 + \frac{1}{(1-\delta)^2}\frac{1}{V^2} 
 \sum_{k\in \Lambda^*_+} \frac{1}{e_k^2}
 - \frac{1}{V^2}\sum_{k,k'\in \Lambda^*_+} \frac{1}{e_{k'} e_k}
 \\ & \quad \le
  \frac{2\delta }{(1-\delta)^2} \rhoc^2 + \frac{1}{(1-\delta)^2}\frac{1}{V^2} 
 \sum_{k\in \Lambda^*_+} \frac{1}{e_k^2}
 \le \frac{\rho^2}{(1-\delta)^2} 
   \tilde{\delta}
  \,,
 \end{align*}
using the definition in (\ref{eq:deftildedelta}) and the assumption $\rho>\rhoc$.  Together with the earlier estimates this completes the proof of the Lemma.
\end{proof}

\begin{corollary}\label{th:mainthgbounds}
 If $\delta\le \frac{1}{2}$ and $\tilde{\delta}\le \frac{\Delta^2}{2^4 V_0^2\rho^2}$, then 
 $Z_1\le 2 Z_0$, $\mean{(\alpha')^2}_{\mu_1} \le 4 \rho^2 \Delta^{-2}\tilde{\delta}$, and 
 \begin{align}\label{eq:gnearonebound2} 
 0\le g[\Phi]\le 2 \left(\frac{\rho}{\Delta}\right)^{V_0-1}\cf{\rho'[\Phi] < \rho}\,, \qquad \mean{|1-g|^2}_{\mu_1} \le 4 
 \left(\frac{\rho}{\Delta}\right)^{2 V_0} \left(1+2^4 V_0^2\right)\tilde{\delta}\,.
\end{align}
\end{corollary}
The assumptions made in the Theorem indeed guarantee that $\delta\le \frac{1}{2}$ and 
$\tilde{\delta}\le \frac{\Delta^2}{2^4 V_0^2\rho^2}$, since $\tilde{\delta}\le 2 \vep$.  Hence, we may continue the proof of the Theorem assuming that all of the conclusions in Corollary \ref{th:mainthgbounds} are valid.

The above representation allows to construct a coupling $\gamma$ between $\mu_0$ and $\mu_1$ by combining the change of variables $G$ with the diagonal concentration trick mentioned earlier.  Together with the estimates in  Corollary \ref{th:mainthgbounds}
this will prove the 
bound stated for the Wasserstein distance between $\mu_0$ and $\mu_1$ in the Theorem.
Explicitly, we define a positive Borel measure $\gamma$ by its action on bounded continuous functions $F(\Phi,\Psi)$¸ as follows:
\begin{align}\label{eq:defmaingamma}
& \mean{F}_{\gamma} := 
\int \mu_1[\rmd \Phi] \min(1,g[\Phi]) F(\Phi,G[\Phi]) 
\nonumber \\ & \quad
+ \int \mu_1[\rmd \Phi] \int \mu_1[\rmd \Psi] 
\frac{1}{Z'} (1-g(\Phi))_+ (g(\Psi)-1)_+   F(\Phi,G[\Psi]) 
\,.
\end{align}
Here $(r)_+:= r \cf{r>0}$ and the normalization factor $Z'$ is given by 
\[
 Z' := \mean{(1-g)_+}_{\mu_1} = \mean{(g-1)_+}_{\mu_1} = \frac{1}{2}\mean{|g-1|}_{\mu_1}\,,
\]
where the second equality follows from the identity $g=1+(g-1)_+-(1-g)_+$ and the earlier made observation that 
$\mean{g}_{\mu_1}=1$ by (\ref{eq:mu0frommu1}). The final equality is then a consequence of the identity 
$|g-1|=(g-1)_+ +(1-g)_+$.
If $f$ is bounded and continuous and $F(\Phi,\Psi)=f(\Phi)$, a straightforward computation shows that $\mean{F}_\gamma = \mean{f}_{\mu_1}$.
If $F(\Phi,\Psi)=f(\Psi)$, a similar computation and using the representation in (\ref{eq:mu0frommu1}) proves that
$\mean{F}_\gamma = \mean{f}_{\mu_0}$.  Therefore, $\gamma$ is indeed a coupling between $\mu_0$ and $\mu_1$.

Using this coupling, we can now conclude that
\begin{align*}
 & W_p(\mu_1,\mu_0)^p \le
\int \mu_1[\rmd \Phi] \min(1,g[\Phi]) \norm{\Phi-G[\Phi]}^p
\nonumber \\ & \quad
+ \int \mu_1[\rmd \Phi] \int \mu_1[\rmd \Psi] 
\frac{1}{Z'} (1-g(\Phi))_+ (g(\Psi)-1)_+ \norm{\Phi-G[\Psi]}^p \,.
\end{align*}
In particular, in the case $p=2$, we can simplify the computations by first using 
 the upper bound $\norm{\Phi-G[\Psi]}^2\le 2 \norm{\Phi-\Psi}^2+ 2\norm{\Psi-G[\Psi]}^2$, which shows that
\begin{align*}
 & W_2(\mu_1,\mu_0)^2 \le 2
\mean{g[\Phi]\,\norm{\Phi-G[\Phi]}^2}_{\mu_1}
\nonumber \\ & \qquad
+ 2 \int \mu_1[\rmd \Phi] \int \mu_1[\rmd \Psi] 
\frac{1}{Z'} (1-g(\Phi))_+ (g(\Psi)-1)_+ \norm{\Phi-\Psi}^2 
\,.
\end{align*}

Let us begin with the second term on the right hand side.  The integrand is zero unless $g(\Psi)>1$.  In particular, then we must have  $\rho'[\Psi]<\rho$,
implying that $\norm{\Psi^+}^2 = V \rho_+[\Psi]\le V \rho'[\Psi]<V \rho $.  On the other hand, under the measure $\mu_1$,
it holds almost surely that $\norm{\Psi^0}^2=V \Delta$.  Therefore, almost surely in the above integrand
\[
 \norm{\Phi-\Psi}^2 \le 2 (\norm{\Phi}^2+\norm{\Psi}^2) \le 2 (\norm{\Phi}^2+V \Delta + V\rho) \,.
\]
Taking into account the definition of $Z'$, we find an estimate
\begin{align*}
 &  \int \mu_1[\rmd \Phi] \int \mu_1[\rmd \Psi] 
\frac{1}{Z'} (1-g(\Phi))_+ (g(\Psi)-1)_+ \norm{\Phi-\Psi}^2 \\
&\quad
\le 2  \int \mu_1[\rmd \Phi] 
(\norm{\Phi}^2+V \Delta + V\rho)(1-g(\Phi))_+\\
&\quad
\le 2  \left[\int \mu_1[\rmd \Phi] \norm{\Phi}^2 |1-g(\Phi)|  + 
V (\rho+\Delta) \int \mu_1[\rmd \Phi] |1-g(\Phi)|  \right]\\
&\quad
\le 2 \left(\mean{\norm{\Phi}^4}_{\mu_1}^{\frac{1}{2}}+V(\rho+\Delta)\right)
 \mean{(1-g)^2}_{\mu_1}^{\frac{1}{2}}\,. 
\end{align*}
Using the definitions, we find that $\norm{\Phi}^2=\norm{\Phi^+}^2+\norm{\Phi^0}^2$.  Therefore,
\begin{align*}
 & \mean{\norm{\Phi}^4}_{\mu_1} =
 \mean{(\norm{\Phi^+}^2+\norm{\Phi^0}^2)^2}_{\mu_1}  \le 2\left(\mean{\norm{\Phi^+}^4}_{\mu_1} +V^2 \Delta^2\right)
\end{align*}
and using the expectations computed in (\ref{eq:phievenexp})
\begin{align*}
 & \mean{\norm{\Phi^+}^4}_{\mu_1} =\int_{\Lambda_+^*}\!\rmd k_1\int_{\Lambda_+^*}\!\rmd k_2\,
 \mean{|\Phi^+(k_1)|^2 |\Phi^+(k_2)|^2 } \\ & \quad
  = V^{-2} \sum_{k,k'\in \Lambda_+^*,\,k'\ne k} \frac{V^2}{e_k e_{k'}} + V^{-2}\sum_{k\in \Lambda_+^*} 2 \frac{V^2}{e_k^2}
  \le 2 V^2 \rhoc^2 \, .
\end{align*}
By assumption, this term is bounded by $2 V^2 \rho^2$, and we may conclude that 
\[
 \mean{\norm{\Phi}^4}_{\mu_1} \le 2\left( 2 V^2 \rho^2+V^2 \Delta^2\right) \le 2^2 V^2 (\rho+\Delta)^2\,.
\]
Therefore, 
\begin{align*}
 & W_2(\mu_1,\mu_0)^2 \le 2 \mean{g[\Phi]\norm{\Phi-G[\Phi]}^2}_{\mu_1}
+ 12 (\rho+\Delta) L^d \mean{(1-g)^2}_{\mu_1}^{\frac{1}{2}}\,.
\end{align*}

By Corollary \ref{th:mainthgbounds}, $\mean{(1-g)^2}_{\mu_1}^{\frac{1}{2}}\le 2^3 V_0 (\rho/\Delta)^{V_0} \sqrt{2 \tilde{\delta}}$, and thus
the second term is  bounded by a constant $3\cdot 2^6 (\rho+\Delta) V_0 (\rho/\Delta)^{V_0}$ 
times $L^d \sqrt{\vep}$.  In addition, using the definition (\ref{eq:defGchange}) and Corollary \ref{th:mainthgbounds}, we find for the first term
\begin{align*}
 & 2 \mean{g[\Phi]\norm{\Phi-G[\Phi]}^2}_{\mu_1}  \le  4 \left(\frac{\rho}{\Delta}\right)^{V_0-1}\int_{\Lambda_+^*}\!\rmd k\, \bigmean{\cf{\rho'[\Phi] < \rho}
             \Bigl[1-\Bigl(1-\frac{\tilde{\vep}[\Phi]}{e_{k}}\Bigr)^{-\frac{1}{2}}\Bigr]^2 |\Phi_k|^2}_{\!\!\mu_1}\\
            & \quad + 4 \left(\frac{\rho}{\Delta}\right)^{V_0-1} \int_{\Lambda_0^*}\!\rmd k\, \bigmean{\cf{\rho'[\Phi] < \rho}
              \Bigl[1-\Bigl(1-\alpha'[\Phi]\Bigr)^{\frac{1}{2}}\Bigr]^2 |\Phi_k|^2}_{\!\mu_1}\,.
\end{align*}
Here, whenever $\rho'[\Phi] < \rho$ and $k\in \Lambda_+^*$,
we may use the identity $1-1/\sqrt{c}=(c-1)/(c+\sqrt{c})$, valid for all $c> 0$, and definition of the relative energy gap, to estimate
\[ 
 \Bigl[1-\Bigl(1-\frac{\tilde{\vep}}{e_{k}}\Bigr)^{-\frac{1}{2}}\Bigr]^2\le
 \frac{\tilde{\vep}^2}{e_{k}^2} \frac{1}{1-\frac{\tilde{\vep}}{e_{k}}} \le\frac{\delta^2}{1-\frac{\tilde{\vep}}{e_{k}}} \,.
\]
Therefore,
\begin{align*}
 & \int_{\Lambda_+^*}\!\rmd k\, \bigmean{\cf{\rho'[\Phi] < \rho}
             \Bigl[1-\Bigl(1-\frac{\tilde{\vep}[\Phi]}{e_{k}}\Bigr)^{-\frac{1}{2}}\Bigr]^2 |\Phi_k|^2}_{\!\!\mu_1}
           \\ &\quad   \le
     \delta^2 
     \bigmean{\cf{\rho'[\Phi] < \rho} \int_{\Lambda_+^*}\!\rmd k\, \frac{e_{k}}{e_{k}-\tilde{\vep}[\Phi]}  |\Phi_k|^2}_{\!\!\mu_1}
     =     \delta^2 V
     \bigmean{\cf{\rho'[\Phi] < \rho} \rho'[\Phi]}_{\!\!\mu_1} \le \rho \delta^2 L^d\,.
\end{align*}

Similarly, we have $1-\sqrt{c}=(1-c)/(1+\sqrt{c})$ for all $c\ge 0$, and thus
\[ 
\Bigl[1-\Bigl(1-\alpha'\Bigr)^{\frac{1}{2}}\Bigr]^2 \le |\alpha'|^2\,.
\]
Since the weight is the same for all components $k\in \Lambda_0^*$,
we find using Corollary \ref{th:mainthgbounds}
\begin{align*}
 & 
 \int_{\Lambda_0^*}\!\rmd k \bigmean{\cf{\rho'[\Phi] < \rho}
              \Bigl[1-\Bigl(1-\alpha'[\Phi]\Bigr)^{\frac{1}{2}}\Bigr]^2 |\Phi_k|^2}_{\!\!\mu_1} 
           \le 
          \bigmean{\cf{\rho'[\Phi] < \rho}|\alpha'[\Phi]|^2 V \rho_0[\Phi]}_{\!\mu_1} 
          \le  4 \Delta^{-1} \rho^2 L^d \tilde{\delta}\,.
\end{align*}

Therefore, since $\delta\le \frac{1}{2}$ and $\delta\le \frac{\vep}{2}$, we can add up and simplify the above bounds to arrive at the bound
\[
 2 \mean{g[\Phi]\norm{\Phi-G[\Phi]}^2}_{\mu_1}  \le 2^5 (\rho+\Delta) (\rho/\Delta)^{V_0} L^d \vep\,.
\]
The assumptions about $\vep$ allow to simplify this slightly to make the weight comparable to that of the first term.  Namely, since now $\sqrt{\vep}\le \Delta/(4 \rho)\le 2^{-2}$, we have proven that
\begin{align*}
 & W_2(\mu_1,\mu_0)^2 \le \left(2^3 (\rho+\Delta) (\rho/\Delta)^{V_0} + 3\cdot 2^6 (\rho+\Delta) V_0 (\rho/\Delta)^{V_0}\right) L^d \sqrt{\vep}
 \\ & \quad
 \le  2^8 (\rho+\Delta) V_0 (\rho/\Delta)^{V_0} L^d \sqrt{\vep}\,.
\end{align*}
Taking the square root, we conclude that the claim in the Theorem follows from the assumptions
for the measure $\mu_1$ defined in (\ref{eq:defmu1inproof}) and the explicit form for the constant $C_2$ stated in the Theorem.
\end{proofof}

\begin{proofof}{Proposition \ref{th:mu1prop}, item \ref{it:smalldev}:}
  If $e_k=0$ for all $k\in \Lambda^*_0$, we are back to the case in item \ref{it:allground}, and since then $\mu'_1=\mu_1$, its conclusions imply also the conclusions of item \ref{it:smalldev} whenever $0\le \tilde{\vep}\le 1$.
  
  Suppose thus that there is some $k\in \Lambda^*_0$ for which $e_k>0$ and that there is
  $\tilde{\vep}\le 1$ for which 
  $e_k\le \frac{1}{2 \rho}L^{-d}\tilde{\vep}$ for all $k\in \Lambda_0^*$.  
  Clearly, then $\tilde{\vep}>0$.
 Comparing the definitions of $\mu_1$ and $\mu'_1$, we have $\mu_1[\rmd \Phi]=g_1(\Phi)\mu'_1[\rmd \Phi]$
 for
 \[
  g_1(\Phi):= \frac{Z'_1}{Z_1} g_2(\Phi)\,, \qquad g_2(\Phi) := \rme^{-E_0[\Phi]\left(1-\frac{\rhoc}{\Delta}\right)}
 \prod_{k\in \Lambda^*_+} \left(1-\frac{E_0[\Phi]L^{-d}}{e_{k} \Delta} \right)^{-1}\,.
 \]
 Here $g_2$ depends only on $\Phi^0$ and satisfies $\mean{g_2}_{\mu_1'}=\frac{Z_1}{Z'_1}$.
 
 As before, the assumptions are tailored to guarantee that $g_1$ remains close to one, and then an explicit good coupling can be found between  $\mu_1$ and $\mu'_1$.  As the small parameter we use here
 \[
  \delta' := \rho V \max_{k\in \Lambda^*_0} e_k\le   \frac{1}{2}\tilde{\vep}\le \frac{1}{2}\,.
 \]
 In particular, we now have almost surely under $\mu_1'$
\[
0\le E_0[\Phi]\le \max_{k\in \Lambda_0^*}e_k N_0[\Phi] = V \Delta \frac{\delta'}{\rho V }= \delta'\frac{\Delta }{\rho}
\le \delta'\,.
\]
Since $-\ln(1-c)\le 2 c$ for $0\le c\le \frac{1}{2}$, we find using the earlier assumption $\delta\le \frac{1}{2}$ that
almost surely under $\mu_1'$ 
\[
 0\le -\ln  \left(1-\frac{E_0[\Phi]L^{-d}}{e_{k} \Delta} \right)\le 2\frac{E_0[\Phi]L^{-d}}{e_{k} \Delta}\le 2 \delta'\frac{1}{\rho V e_{k}}  \,,
\]
for all $k\in \Lambda^*_+$.  Therefore,
\[
0\le \sum_{k\in \Lambda^*_+}\ln  \left(1-\frac{E_0[\Phi]L^{-d}}{e_{k} \Delta} \right)^{-1}\le 2 \delta'\frac{\rhoc}{\rho}\le 2 \delta'
 \,.
\]
Similarly, $E_0[\Phi]\frac{\rhoc}{\Delta}\le \delta'$, and thus we have obtained almost sure bounds
\[
 \rme^{-\delta'}\le g_2(\Phi)\le \rme^{3 \delta'}\,.
\]
Taking expectation over $\mu'_1$ we find also that
\[
 \rme^{-\delta'} \le \frac{Z_1}{Z'_1}\le \rme^{3 \delta'}\,.
\]
Combining these two results shows that almost surely under $\mu_1'$ 
\[
 \rme^{-4 \delta'}\le g_1(\Phi)\le \rme^{4 \delta'}\,.
\]
Since $\delta'\le \frac{1}{2}$, this yields an almost sure bound
\begin{align}\label{eq:aeoneminusg1}
 |1-g_1(\Phi)|\le \rme^{4 \delta'}|1-\rme^{-4 \delta'}|\le 4 \rme^{2} \delta'\,. 
\end{align}

We define a measure $\gamma_1$ by setting for bounded continuous functions $F(\Phi,\Psi)$
\begin{align}\label{eq:defgammaone}
& \mean{F}_{\gamma_1} := 
\int \mu'_1[\rmd \Phi] \min(1,g_1[\Phi]) F(\Phi,\Phi) 
\nonumber \\ & \quad
+ \int \mu'_1[\rmd \Phi] \int \mu'_1[\rmd \Psi] 
\frac{1}{Z''} (1-g_1(\Phi))_+ (g_1(\Psi)-1)_+   F(\Phi,\Psi) 
\,.
\end{align}
where 
\[
 Z'' := \mean{(1-g_1)_+}_{\mu'_1} = \mean{(g_1-1)_+}_{\mu'_1}\,.
\]
Note that, since $E_0$ is not a constant function, $g_1$ cannot be a constant function, and hence $Z''>0$.
As before, it is then straightforward to check that the first marginal equals $\mu'_1$ and the second marginal equals $\mu_1$.

Therefore, $\gamma_1$ is a coupling  between $\mu_1$ and $\mu'_1$, and we have
\[
  W_2(\mu_1,\mu'_1)^2 \le \int \mu'_1[\rmd \Phi] \int \mu'_1[\rmd \Psi] 
\frac{1}{Z''} (1-g_1(\Phi))_+ (g_1(\Psi)-1)_+   \norm{\Phi-\Psi}^2\,.
\]
Again, we estimate $\norm{\Phi-\Psi}^2 \le 2 (\norm{\Phi}^2+ \norm{\Psi}^2)$, and use the symmetry and definition of $Z''$ to obtain a bound
\[
  W_2(\mu_1,\mu'_1)^2 \le 2 \mean{\norm{\Phi}^2 |1-g_1(\Phi)|}_{\mu_1'} \,.
\]
Combined with the almost sure bound in (\ref{eq:aeoneminusg1}), we find that 
\[
  W_2(\mu_1,\mu'_1)^2 
  \le 2^3 \rme^{2} \delta'\mean{\norm{\Phi}^2}_{\mu_1'} \,.
\]
Here, $\mean{\norm{\Phi}^2}_{\mu_1'} = \mean{\norm{\Phi^+}^2+\norm{\Phi^0}^2}_{\mu_1'} = V \rhoc + V \Delta=V \rho$.
Therefore, 
\[
 W_2(\mu_1,\mu'_1)^2 \le 2^5  V \rho \tilde{\vep}\,.
\]
Note that we obtained a better dependence on $\tilde{\vep}$ than on $\vep$ in the earlier estimate since we did not need to use the Schwarz inequality above.  This was possible here since the weight $g_1$ is almost surely close to one unlike the weight $g$ which is close to one only with high probability.

Since the Wasserstein metric satisfies the triangle inequality, we can now combine the above bound with the one proved in Theorem \ref{th:mainresult}, and conclude that
\[
 W_2(\mu_0,\mu'_1) \le W_2(\mu_0,\mu_1)+W_2(\mu_1,\mu'_1)+ \le
 L^{\frac{d}{2}} 2^4 \sqrt{V_0 (\rho+\Delta)} 
 \left((\rho/\Delta)^{\frac{V_0}{2}}\vep_L^{\frac{1}{4}}+  \tilde{\vep}^{\frac{1}{2}}\right)\,,
\]
as claimed in the Proposition.
\end{proofof}

\section{Proof of the existence of the energy gap, Lemma \ref{th:mainomlemma}}\label{sec:proofofgap}

Here we suppose $d\ge 3$ and consider a dispersion relation
$\omega$ which satisfies Assumption \ref{as:omega}.  For each $L$, define $\omega_0$
and $e_k$, $k\in\Lambda^*$, as in Definition \ref{th:defslpit}.
We choose $\kappa$ such that $0<\kappa<\frac{d}{2}$, if $d\ge 4$, and $0<\kappa < 1$, if $d=3$, and fix its value for the rest of the proof.
In principle, only the local behaviour of $\omega$ around its global minima will matter,
but the proof is complicated by the fact that the local behaviour in a neighbourhood of each minima can be different and the values of $e_k$ can become mixed between the minima.

The proof will be composed out of several steps.  The steps are not completely independent, 
and each step may use estimates and notations accumulated from the previous steps.
Although the proof is not isolated into technical Lemmas, 
the steps highlight its structure by each having a specific goal, listed in the following:
\begin{enumerate}
 \item\label{it:continneighb} Isolate sufficiently small neighbourhoods in $\T^d$ around each 
 minimum of $\omega$ so that second order Taylor series bounds its behaviour in the neighbourhood.
 \item\label{it:largeL0}  Choose sufficiently large $L$ so that the rectangular grid $\Lambda^*(L)$ has some points in each neighbourhood.
 \item\label{it:defLambda1}  Construct a condensate candidate set $\Lambda^*_1$
 by isolating all small energies, with an energy difference from the lowest energy  proportional to $L^{-2}$.  Show that the number of points in this set is bounded by some $M_0$ which does not depend on $\kappa$ nor on $L$
 \item\label{it:defLambda0}  Use a ``pigeon hole'' argument to show 
 that this set must contain a large enough relative energy gap.  This will fix 
 the condensate wave number set $\Lambda^*_0$, hence also $\Lambda^*_+$, and complete the proof of item \ref{it:splitcondensate} of the Lemma. 
 \item\label{it:item2} Check that the relative energy gap of the construction satisfies item
  \ref{it:splitgap} of the Lemma. 
 \item\label{it:item3a} Use the previous estimates to find a constant $c_2$ for the bound 
 (\ref{eq:e2sumbound}), separately for $d=3$, $d=4$, and $d\ge 5$.
 \item\label{it:Riemann}  Using an approximation with suitable Riemann sums, prove the estimates (\ref{eq:defrhoinfty}) and (\ref{eq:rhoerrorbound}) for the continuum limit $L\to \infty$.
\end{enumerate}

{\it (Step \ref{it:continneighb})}
Consider a point $k_0\in T_0$ where $\omega(k_0)=\ommin$.
Since $k_0$  is a non-degenerate minimum of a twice continuously differentiable function $\omega$, we have $\nabla \omega(k_0)=0$
and the eigenvalues of $D^2 \omega(k_0)$ are strictly positive.  Let $\lambda_-$ and $\lambda_+$ denote the smallest and, respectively, the largest of these eigenvalues
as $k_0$ varies through the elements in $T_0$.  Then $0<\lambda_-\le \lambda_+$.
By continuity of $D^2 \omega$ there is $\delta>0$ such that $\delta<\frac{1}{2}$, and 
whenever\footnote{We make a slight abuse of notations here: By ``$|k-k_0|$'' we mean $d_{\T^d}(k,k_0)$, where $d_{\T^d}$ is the periodic distance on the torus, inherited as a quotient metric from the definition $\T^d=\R^d/\Z^d$. We are only using this notion for distances which are less than one half, in which case there is a metric isomorphism between a ball in $\R^d$ and an open subset of the torus containing the geodesic line connecting the points $k$ and $k_0$.  In this case, the metric behaves as the norm in $\R^d$, and the notation should not be overly misleading.} $k_0\in T_0$, $|k-k_0|<\delta$, and $p\in\R^d$ we have
\[
 \frac{\lambda_-}{2}|p|^2 < p\cdot (D^2 \omega(k)p) < 2 \lambda_+ |p|^2 \,.
\]
As $T_0$ is finite, we can also assume that the balls $B(k_0,\delta)$ are disjoint, by choosing a smaller $\delta$
if this is not true initially.
Since the set $\defset{k\in \T^d}{|k-k_0|\ge \delta,\text{ for all }k_0\in T_0}$ is compact, the continuous function $\omega$ has a minimum value $\omega_2$ which is attained within the set.  
Then we must have $\omega_2>\ommin$ since else the point $k$ at which $\omega(k)=\omega_2$ would belong to $T_0$.
Furthermore, by a Taylor expansion up to second order around $k_0$, we find that if $k_0\in T_0$ and $|k-k_0|<\delta$,
then
  \begin{align}\label{eq:locallowerom}
    \frac{\lambda_-}{4} |k-k_0|^2 \le \omega(k)-\ommin \le \lambda_+  |k-k_0|^2\,,\qquad |\nabla\omega(k)| \le 2 \lambda_+ |k-k_0| \,.
  \end{align}

{\it (Step \ref{it:largeL0})}
We are going to define a cut-off size $L_0$, and consider lattices with $L\ge L_0$.
We begin by assuming that $L_0 \in \N_+$ satisfies
\begin{align}\label{eq:L0assumptions}
 L_0>\frac{\sqrt{d}}{2 \delta}\,,\quad L_0\ge \left[\frac{c_0}{\omega_2-\ommin}\right]^{\frac{1}{2}}\,,
\end{align}
where $c_0$ is an $L$-independent constant depending on $\omega$ via $\lambda_+$,
\begin{align}\label{eq:c0def}
 c_0 := \frac{\lambda_+ d}{2}\,. 
\end{align}
For any such $\Lambda^*(L)$, 
let us first isolate the minimum value of $\omega$ on these points,
i.e., set as in the Lemma
\[
 \omega_0(L) := \min_{k\in \Lambda^*} \omega(k)\,.
\]
As shown by the examples in Sec.~\ref{sec:examples}, $\omega_0$ may then depend on $L$, and even if $\omega$ would have more than one minimum 
point on $\T^d$, the value of $\omega_0$ could be unique.  

Since $\Lambda^*$ forms a rectangular grid with side length 
$\frac{1}{L}$ on $\T^d$, to any point $k\in \T^d$ there is a point $k'\in \Lambda^*$ such that $|k - k'|_\infty \le \frac{1}{2 L}$.
Since $|p|_\infty = \max_i |p_i|\ge d^{-\frac{1}{2}} |p|$, then $|k - k'| \le \frac{\sqrt{d}}{2 L}\le \frac{\sqrt{d}}{2 L_0}< \delta$.
Therefore, if $k_0\in T_0$, there is $k_0'\in \Lambda^*$ for which $|k'_0 - k_0|\le \frac{\sqrt{d}}{2 L} <\delta$, and thus
$\omega(k'_0)-\ommin \le \lambda_+ |k'_0-k_0|^2 \le \frac{\lambda_+ d}{4} L^{-2}$.
This implies that
\[
 0\le \omega_0(L)-\ommin \le \frac{c_0}{2} L^{-2}\,.
\]
In particular, $\omega_0(L)\to \ommin$ as $L\to \infty$.  

{\it (Step \ref{it:defLambda1})}
We recall that $e_k = \omega(k)-\omega_0$ for $k\in \Lambda^*$,
and consider the following set of $k$ which have an energy close to the ground state:
\begin{align}
 \Lambda^*_1 := \defset{k\in\Lambda^*}{e_k< \frac{c_0}{2} L^{-2}}\,.
\end{align}
Clearly, any minimum point has $\omega(k)=\omega_0$ and thus it belongs to $\Lambda^*_1$.  Hence, $\Lambda^*_1$ 
is not empty.
In addition, the second inequality in (\ref{eq:L0assumptions}) implies that if $k\in \Lambda^*_1$, then 
$\omega(k)-\ommin = e_k  + \omega_0-\ommin < c_0 L^{-2}\le \omega_2-\ommin$.  Therefore, to each
$k\in \Lambda^*_1$, we can find a unique $k_0\in T_0$ such that $|k-k_0|<\delta$ and the inequalities (\ref{eq:locallowerom}) hold. 

For each $k_0\in T_0$, let us next consider the values in the subset
\begin{align}
 \Lambda^*(k_0;L) := \defset{k\in\Lambda^*}{|k-k_0|< \delta}\,.
\end{align}
By the same reasoning as above, we can find $n_0\in \Z^d$ for which $|n_0-L k_0|_\infty \le \frac{1}{2}$.
Therefore, is it possible to reparameterize the values in $\Lambda^*(k_0;L)$ defining 
$m(k)=(L k -n_0)\bmod \Lambda_L$ for each $k\in \Lambda^*(k_0;L)$.  Note that then for all $k\in \Lambda^*(k_0;L)$
we have $L k=(n_0+m(k))\bmod \Lambda_L$ and 
$L|k-k_0|_\infty = L \inf_{n\in \Z^d}|k-k_0-n|_\infty= |m(k)+n_0-L k_0|_\infty \ge |m(k)|_\infty - \frac{1}{2}$.
On the other hand, if $k\in \Lambda^*(k_0;L)\cap \Lambda^*_1$, 
\[
 \frac{\lambda_-}{4} |k-k_0|_\infty^2 \le \frac{\lambda_-}{4} |k-k_0|^2 \le \omega(k)-\ommin < c_0 L^{-2}\,,
\]
and thus also 
\[
 L|k-k_0|_\infty \le \sqrt{\frac{4 c_0}{\lambda_-}}\,.
\]
Therefore, then $|m(k)|_\infty \le \frac{1}{2}+\sqrt{\frac{2 \lambda_+ d}{\lambda_-}}$.
We define
\begin{align}\label{eq:defMmax}
 M := \left\lfloor \frac{1}{2}+\sqrt{\frac{2 \lambda_+ d}{\lambda_-}} \right\rfloor \,,
\end{align}
where $\lfloor x\rfloor$ denotes the smallest integer in $\Z$ less than or equal to $x\in \R$.  Then $M\ge 0$, and there are at most $(2 M+1)^d$ values $m\in \Z^d$ which can satisfy $|m|_\infty \le M$.
Even if the maximal number of points occur in $\Lambda^*(k_0;L)\cap \Lambda^*_1$ at each $k_0\in T_0$, we conclude that there are at most 
\begin{align}\label{eq:defmaxN0}
 M_0 := |T_0| (2 M+1)^d
\end{align}
points in $\Lambda^*_1$.  

{\it (Step \ref{it:defLambda0})}
We are  next going to construct $\Lambda^*_0$ as a subset of $\Lambda^*_1$, and then also $|\Lambda^*_0|\le M_0$ and 
$0\le \omega(k)-\ommin < c_0 L^{-2}$ for all $k\in \Lambda^*_0$.  
Let us stress that $M_0$ is indeed independent of $L$ and $\kappa$, as required in the Lemma.  For simplicity, we now add one more requirement for $L_0$: we assume that $L_0^d\ge M_0+1$, so that if $L\ge L_0$, the complement of $\Lambda^*_1$ cannot be empty.

To isolate those Fourier modes which behave as a condensate, recall that $\kappa$ has been fixed to satisfy the requirements of the Lemma.  Define $b'_L= \frac{1}{2}c_0 L^{-d+\kappa}$ and $r_L:= L^{-\frac{d-2-\kappa}{M_0}}$, to denote the two bounds appearing in item \ref{it:splitgap} of the Lemma.
Then $r_L\le 1$, since $L\ge 1$, and the assumptions imply that $\kappa< d-2$.
We also have 
\[
L^2 b'_L  = \frac{1}{2}c_0 L^{-d+\kappa+2}= \frac{1}{2}c_0 r_L^{M_0}\le \frac{1}{2}c_0\,.
\]
 Therefore, if $e_k<b'_L$, also $e_k<\frac{c_0}{2}L^{-2}$,
and thus $k\in \Lambda^*_1$.   All of these values of $k$ will be included in $\Lambda^*_0$ but to find a suitable gap, we might need to include 
also some values from the remainder set,
\[
 \Lambda^*_2 := \defset{k\in \Lambda_1^*}{e_k \ge b'_L} = 
 \defset{k\in \Lambda^*}{b'_L\le e_k < \frac{c_0}{2} L^{-2}}\,.
\]

If $\Lambda^*_2=\emptyset$, we can conclude that $e_k<b'_L$ for each $k\in \Lambda^*_1$ and, if $k'\in \Lambda^* \setminus \Lambda^*_1$,
we have $e_{k'} \ge \frac{c_0}{2} L^{-2} = r_L^{-M_0} b'_L \ge r_L^{-1} b'_L > r_L^{-1}e_k$.
Therefore, we may then define $\Lambda^*_0=\Lambda^*_1$ and the corresponding split is separated by $[a_L,b_L]$ and 
has an energy gap $\delta_L^{-1}$, where
$\delta_L<r_L$, $a_L:=b'_L$, $b_L:= r_L^{-M_0} b'_L\ge a_L$.

Suppose thus that $N_2:=|\Lambda^*_2|>0$, and enumerate the elements $k_i\in \Lambda^*_2$, $i=1,2,\ldots,N_2$, so that 
$o_i = e_{k_i}$ form an increasing sequence, $o_{i+1}\ge o_i$ for all $i$.  Define also $o_{N_2+1} := \min_{k\in \Lambda^* \setminus \Lambda^*_1} e_k
\ge \frac{c_0}{2} L^{-2}$ and $o_0 := \max_{k\in \Lambda^*_1 \setminus \Lambda^*_2} e_k< b'_L$.  Note that at least all minimum points belong to $\Lambda^*_1 \setminus \Lambda^*_2$ and our $L$ is large enough so that $\Lambda^* \setminus \Lambda^*_1$ cannot be empty.
Clearly, also the new sequence of $o_i$, $i=0,1,\ldots,N_2+1$, is increasing.
Therefore, we can use a pigeon hole argument to the relative energies: We have
\[
  (N_2+1) \max_{i=0,1,\ldots, N_2} \ln \frac{o_{i+1}}{o_i}  
 \ge \sum_{i=0}^{N_2} \ln \frac{o_{i+1}}{o_i} = \ln \left(\prod_{i=0}^{N_2} \frac{o_{i+1}}{o_i}\right) =
  \ln \left(\frac{o_{N_2+1}}{o_0}\right) \ge
  \ln \left(\frac{c_0}{2 L^2 b'_L}\right)\,.
\]
The right hand side is equal to $\ln r_L^{-M_0}=M_0 \ln r_L^{-1}$, and
since $N_2+1 \le |\Lambda^*_1|\le M_0$, there is at least one $i\in\set{0,1,\ldots, N_2}$ for which
\[
 \frac{o_{i+1}}{o_i} \ge r_L^{-1}\, .
\]
Let $j$ denote the smallest of such $i$, and define
\[
 \Lambda^*_0 := \defset{k\in \Lambda^*}{e_k \le o_j}\,.
\]
By construction, $o_0\le o_j< \frac{c_0}{2} L^{-2}$ and thus $\Lambda^*_1\setminus \Lambda^*_2\subset \Lambda^*_0\subset \Lambda^*_1$.  
Therefore, neither $\Lambda^*_0$ nor its complement $\Lambda^*_+$ can be empty, and $|\Lambda^*_0|\le M_0$.
In addition, $0\le \omega(k)-\ommin < c_0 L^{-2}$ for all $k\in \Lambda^*_0$, and thus   
$(\Lambda_0^*,\Lambda^*_+)$ forms a split of $\Lambda^*$ which satisfies item \ref{it:splitcondensate} of the Lemma.

{\it (Step \ref{it:item2})}
In case $j=0$, we have $o_j=o_0<b'_L$.  Otherwise, $j\le N_2\le M_0-1$ and, 
by construction, we have $o_{i+1}< r_L^{-1} o_i$  for all $i<j$.  Since $j\le M_0-1$, we find
\[
 o_j \le r_L^{-j} o_{0} < r_L^{-(M_0-1)} b'_L =  r_L \frac{c_0}{2} L^{-2}\,.
\]
Also by construction, 
if $k'\in \Lambda^*_+$, then $k'\in \Lambda^*_2$ or $k'\in \Lambda^*\setminus \Lambda^*_1$, and in both cases $e_{k'}\ge b'_L$.  Thus we may define $b_L := \min_{k\in \Lambda^*_+} e_k$ for which $b_L\ge b'_L$.
In addition, for any $k\in \Lambda^*_0$ we have
\[
 e_k \le o_j\le r_L o_{j+1} \le r_L e_{k'}\,.
\]
Therefore, setting $a_L:= o_j$, we find that 
this choice results in a split which is separated by $[a_L,b_L]$ and 
has an energy gap $\delta_L^{-1}$, where
$\delta_L\le r_L$.

{\it (Step \ref{it:item3a})}
We have now shown that the split $(\Lambda_0^*,\Lambda^*_+)$ constructed above satisfies 
also item \ref{it:splitgap} of the Lemma, and thus only the bounds stated in item \ref{it:splitbounds} remain to be proven.
We only need to consider values of $e_k$ for $k\in \Lambda^*_+$ for which we have proven a lower bound
$e_k\ge \frac{1}{2}c_0 L^{-d+\kappa}$.  In addition, we may also further divide these values into the sets
\[
 F(k_0) := \Lambda^*(k_0;L)\cap \Lambda^*_+\,, \quad k_0\in T_0\,,
\]
and $F' := \Lambda^*_+ \setminus \left(\cup_{k_0\in T_0} F(k_0)\right)$.  If $k\in F'$, we
have by construction a lower bound $e_k\ge \omega_2-\omega_0$ which by 
(\ref{eq:L0assumptions}) and item \ref{it:splitcondensate} of the Lemma 
is bounded from below by $\omega_2-\ommin-\frac{c_0}{2}L^{-2}\ge \frac{1}{2}\left(\omega_2-\ommin\right)>0$
for all $L\ge L_0$.  Therefore,
\[
 \sum_{k\in F'} \frac{1}{e_k^2} \le \frac{4}{(\omega_2-\ommin)^2} V=O(L^d)\,.
\]

Let us then consider a fixed $k_0\in T_0$ and the values $k\in F(k_0)$.
As explained above, we may parameterize these using integers $m(k)\in \Lambda_L$.
If $|m(k)|_\infty\ge 1$, we have then
$L|k-k_0|_\infty \ge |m(k)|_\infty - \frac{1}{2} \ge \frac{1}{2}|m(k)|_\infty$.
On the other hand, then also
\[
 e_k = \omega(k)-\ommin+\ommin-\omega_0\ge \frac{\lambda_-}{4}|k-k_0|_\infty^2-\frac{c_0}{2}L^{-2}
 \ge \left(\frac{\lambda_-}{2^4}|m(k)|_\infty^2-\frac{c_0}{2}\right)L^{-2}\,.
\]
This implies that whenever $|m(k)|_\infty^2 \ge \frac{2^4 c_0}{\lambda_-}$, we have
$e_k \ge \frac{\lambda_-}{2^5}|m(k)|_\infty^2 L^{-2}$.
For the remaining values we use the bound in item \ref{it:splitgap} of the Lemma, and taking into account that $|m(k)|_\infty\le \frac{L}{2}$,
we may conclude that
\[
 \sum_{k\in F(k_0)} \frac{1}{e_k^2} \le \frac{4}{c_0^2} L^{2 d-2 \kappa} \left(1+2\sqrt{\frac{2^4 c_0}{\lambda_-}}\, \right)^d
 + \sum_{m\in \Z^d} \cf{1\le |m|_\infty\le L/2} L^{4} \frac{2^{10}}{\lambda_-^2} |m|_\infty^{-4}\,.
\]

The remaining sum satisfies a bound
\[
 \sum_{m\in \Z^d} \cf{1\le |m|_\infty\le L/2} |m|_\infty^{-4} \le
 \sum_{n=1}^{L} \frac{1}{n^4} 2 d (2 n+1)^{d-1} \le d 2^{2 d-1} \sum_{n=1}^{L} n^{d-5} \,.
\]
If $d\ge 5$, the terms in the sum over $n$ form an increasing sequence and its value is bounded by $L^{d-4}$.
If $d\le 4$, the summand consists of integer values of the decreasing function $x^{-(5-d)}$.  Thus by a Riemann sum estimate, 
we may use the following bound for $d=4$,
\[
 \sum_{n=1}^{L} n^{-1} \le 1 + \int_1^L\!\rmd s \, \frac{1}{s} = 1+\ln L\,,
\]
and for $d=3$ we obtain
\[
 \sum_{n=1}^{L} n^{-2} \le 1 + \int_1^L\!\rmd s \,\frac{1}{s^2} = 1+1-\frac{1}{L} \le 2\,.
\]

Collecting the above bounds together we find that there is a constant $c>0$, which may vary with $d$ but can be chosen independently of $L$,
such that, if $d=3$, 
\[
 \frac{1}{V} \sum_{k\in \Lambda^*_+} \frac{1}{e_k^2} \le c \left(L^{3-2\kappa}+L\right) \,,
\]
where $3-2 \kappa>1$,
if $d=4$,
\[
 \frac{1}{V} \sum_{k\in \Lambda^*_+} \frac{1}{e_k^2} \le c \left(L^{4-2\kappa}+\ln L+ 1\right) \,,
\]
where $4-2 \kappa>0$, and if $d\ge 5$,
\[
 \frac{1}{V} \sum_{k\in \Lambda^*_+} \frac{1}{e_k^2} \le c \left(L^{d-2\kappa}+1\right) \,,
\]
where $d-2\kappa>0$.  In each of the three cases, the first term in the parenthesis on the right hand side dominates over the second term as $L\to \infty$.  Therefore, 
we can always find a constant $c_2$ so that the bound in (\ref{eq:e2sumbound}) holds for the fixed choice of $\kappa$.

{\it (Step \ref{it:Riemann})}
For the final estimates (\ref{eq:defrhoinfty}) and (\ref{eq:rhoerrorbound}), let us first recall the bounds (\ref{eq:locallowerom}) 
satisfied by $\omega(k)-\ommin$ in a $\delta$-neighbourhood of any of its zeroes.  Using the bounds and spherical coordinates shows that
the integral (\ref{eq:defrhoinfty}) defining $\rho_\infty$ is finite for all $d\ge 3$.  
Denote the integrand by $f(k) := \frac{1}{\omega(k)-\ommin}$ for $k\in \T^d\setminus T_0$, and 
choose arbitrarily $f(k)$ to be zero otherwise.  Suppose that $L\ge L_0$, so that we may use all of the above results, in particular, let us continue to use the split $(\Lambda_0^*,\Lambda^*_+)$ defined above.

Cover $\T^d$ with closed boxes with side length $\frac{1}{L}$
and with $k\in \Lambda^*$ at the centre of each box, i.e., set for each $k\in \Lambda^*$ 
\[
D_k := \defset{k'\in \T^d}{|k'-k|_\infty\le \frac{1}{2L}}\,. 
\]
Clearly, then $\int_{D_k}\!\rmd k'\, 1=L^{-d}$, and thus
\begin{align}\label{eq:Dkcentering}
 \rhoc(L) = L^{-d}\sum_{k\in \Lambda^*_+} \frac{1}{e_k} =
\sum_{k\in \Lambda^*_+} \int_{D_k}\!\rmd k'\, \frac{1}{e_k}  \,. 
\end{align}
On the other hand, the points on the torus which correspond to a point in more than one box form a set of zero measure, 
so we may write
\[
  \rho_\infty = \int_{\T^d}\!\rmd k'\, f(k') = \sum_{k\in \Lambda^*}
  \int_{D_k}\!\rmd k'\, f(k') \,.
\]
Therefore,
\[
  \rho_\infty -  \rhoc(L) = \sum_{k\in \Lambda^*_0}
  \int_{D_k}\!\rmd k'\, f(k') + \sum_{k\in \Lambda^*_+} \int_{D_k}\!\rmd k'\, \left(f(k')-\frac{1}{e_k} \right)\,.
\]

We estimate the error in two parts: First, 
the  sum over $k\in \Lambda^*_0$ and those $k\in \Lambda^*_+$ which are
sufficiently close to some $k_0\in T_0$ can be estimate similarly.  For the remaining 
$k\in \Lambda^*_+$ we use differentiability of $f$ and decay of the error with distance from the singular set $T_0$.

We first recall the above split of $\Lambda^*_+$ into $F'$ and $F(k_0)$, and consider the sum over $k\in F(k_0)$ for some fixed $k_0\in T_0$.
Computing directly from the definitions, we find that
\[
 f(k')-\frac{1}{e_k} = \left(\omega(k)-\omega(k')-\omega_0+\ommin\right) f(k') \frac{1}{e_k}\,.
\]
Here $k'\in D_k$, and thus $|k'-k|_\infty\le \frac{1}{2 L}$.  Hence, by convexity of $D_k$,
\begin{align}\label{eq:omombound}
 |\omega(k)-\omega(k')|\le |k'-k| \sup_{\xi\in D_k} |\nabla\omega(\xi)|\le \frac{1}{L} \frac{\sqrt{d}}{2}\sup_{\xi\in D_k} |\nabla\omega(\xi)|
 \,. 
\end{align}
Using again the parameterization of $k$ by $m(k)$ for which $L|k-k_0|_\infty\le |m(k)|_\infty+ \frac{1}{2}$,
by the second bound in (\ref{eq:locallowerom}) we may estimate for all $\xi\in D_k$ and sufficiently large $L$
\begin{align}\label{eq:omderbound}
 |\nabla\omega(\xi)|\le 2 \lambda_+ |\xi-k_0|\le 2 \lambda_+ \sqrt{d}\, |\xi-k_0|_\infty
 \le 2 \lambda_+ \sqrt{d} \left(\frac{1}{2 L}+\frac{1}{2 L}+\frac{|m(k)|_\infty}{L}\right)\,. 
\end{align}

Therefore, if $k$ is close enough to $k_0$ so that
$|m(k)|_\infty^2 < \frac{2^4 c_0}{\lambda_-}+4$, we can conclude that there is an $L$ and $k$-independent constant $c'$ such that for all $k'\in D_k$ 
\[
|\omega(k)-\omega(k')-\omega_0+\ommin| \le c' L^{-2}\,.
\]
Thus the contribution from such $k$ satisfies
\[
  \int_{D_k}\!\rmd k'\, \left|f(k')-\frac{1}{e_k} \right|\le \frac{c'}{L^2 e_k} \int_{D_k}\!\rmd k'\, f(k')
  \le \frac{2 c'}{c_0} L^{d-2-\kappa} \int_{D_k}\!\rmd k'\, f(k')\,.
\]
In addition, then $|k-k_0|\le \sqrt{d}|k-k_0|_\infty\le L^{-1} c''$, for an $L$-independent constant $c''>0$.
Therefore, the sum of the error terms over these $k$ is bounded by $ \frac{2 c'}{c_0} L^{d-2-\kappa}$ times
\[
\int_{|k-k_0|\le c''/L}\!\rmd k'\, f(k') \le \frac{4}{\lambda_-} |S^{d-1}|\int_0^{c''/L}\!\rmd r \, r^{d-1-2}
= \frac{4}{\lambda_-} |S^{d-1}| \frac{(c'')^{d-2}}{d-2} L^{2-d}\,.
\]
This proves that the error from these terms is $O(L^{-\kappa})$ as $L\to\infty$.  

Since for each $k\in \Lambda^*_0$ we know that $L |k-k_0|_\infty \le \sqrt{4 c_0/\lambda_-}$, an identical argument may be used to conclude that, as $L\to\infty$,
\[
 \sum_{k\in \Lambda^*_0}
  \int_{D_k}\!\rmd k'\, f(k') = O(L^{2-d})=O(L^{-\kappa})\,.
\]

Let us next estimate terms $k\in F(k_0)$ with  $|m(k)|_\infty^2 \ge \frac{2^4 c_0}{\lambda_-}+4$.
By the earlier computations, we know that then $e_k \ge \frac{\lambda_-}{2^5}|m(k)|_\infty^2 L^{-2}$.
On the other hand, since $|m(k)|_\infty \ge 2$,
we also have $|m(k)|_\infty-1 \ge \frac{1}{2} |m(k)|_\infty$, and thus, if $k'\in D_k$, we may estimate 
$|k'-k_0|_\infty \ge |k-k_0|_\infty-|k'-k|_\infty
\ge \frac{1}{L}\left(|m(k)|_\infty-1 \right)\ge  \frac{1}{2 L}|m(k)|_\infty$.
Thus by (\ref{eq:locallowerom})
\[
 \frac{1}{f(k')} = \omega(k')-\ommin \ge \frac{\lambda_-}{4} |k'-k_0|_\infty^2
 \ge \frac{\lambda_-}{2^4} L^{-2} |m(k)|_\infty^{2}\,,
\]
and both $1/e_k$ and $f(k')$ have similar upper bounds.

It is now useful to expand the difference further and integrate the identity
\[
 f(k')-\frac{1}{e_k} = \left(\omega(k)-\omega(k')-\omega_0+\ommin\right)  \frac{1}{e_k^2}
+\left(\omega(k)-\omega(k')-\omega_0+\ommin\right)^2  \frac{1}{e_k^2} f(k')\,.
\]
Since $\int_{D_k}\!\rmd k'\, (k'_i-k_i)=0$ for any $i=1,2,\ldots,d$, we have
\[
 \int_{D_k}\!\rmd k'\,\left(\omega(k)-\omega(k')\right)
 = \int_{D_k}\!\rmd k'\int_0^1\!\rmd \tau \, (1-\tau) (k-k')\cdot D^2\omega(\tau k + (1-\tau) k')(k-k')\,,
\]
and, therefore, 
\[
 \left|\int_{D_k}\!\rmd k'\,\left(\omega(k)-\omega(k')\right)\right| \le \frac{d}{4 L^2} \sup_{\xi\in D_k}\norm{D^2 \omega(\xi)} \frac{1}{2} L^{-d}\,.
\]
Since $\omega$ is  twice continuously differentiable, together with (\ref{eq:omombound})
this shows that there is an $L$-independent constant $C'>0$  such that
\begin{align}\label{eq:eknotsmallbound}
  \left|\int_{D_k}\!\rmd k'\,\left(f(k')-\frac{1}{e_k}\right) \right|
  \le C' L^{-2-d} \frac{1}{e_k^2} + C' \frac{(1+L\sup_{\xi\in D_k} |\nabla\omega(\xi)|)^2}{ L^{4} e_k^2}\int_{D_k}\!\rmd k'\,f(k')\,. 
\end{align}

Therefore, denoting $m=m(k)$, using (\ref{eq:omderbound})
to estimate the derivative, and recalling the earlier upper bounds for $1/e_k$ and $f(k')$, we find that 
\[
  \left|\int_{D_k}\!\rmd k'\,\left(f(k')-\frac{1}{e_k}\right) \right|
  \le C'' L^{2-d} |m|^{-4}\,, 
\]
where the constant $C''$ is independent of $L$.  Estimating the sum over possible values of $m$ as above, we thus find that the contribution from these terms is $O(L^{-1})$, for $d=3$, it is $O(L^{-2}(1+\ln L))$, for $d=4$, and 
$O(L^{-2})$, for $d\ge 5$.  The first two cases are $O(L^{-\kappa})$, and thus we have proven that
\[
 \sum_{k\in \Lambda^*_0}
  \int_{D_k}\!\rmd k'\, f(k') + \sum_{k_0\in T_0} \sum_{k\in F(k_0)} \int_{D_k}\!\rmd k'\, \left(f(k')-\frac{1}{e_k} \right) = O(L^{-\min(\kappa,2)})\,,
\]
as required by the Lemma.

It remains to estimate the contribution from the values with $k\in F'$.  Since then $e_k\ge (\omega_2-\ommin)/2>0$ uniformly in $k$ and $L$, we may simply use the uniform bound for the gradient in (\ref{eq:eknotsmallbound}), and conclude that 
\[
  \sum_{k\in F'} \left|\int_{D_k}\!\rmd k'\,\left(f(k')-\frac{1}{e_k}\right) \right|
  \le C' L^{-2-d} \sum_{k\in F'} \frac{1}{e_k^2} + C''' L^{-2}\sum_{k\in F'}\int_{D_k}\!\rmd k'\,f(k') = O(L^{-2})\,.
\]
Combining all of the above results, we have thus proven that 
\[
   \rhoc(L) = \rho_\infty + O(L^{-\min(\kappa,2)})\,,
\]
which completes the proof of the Lemma.

\appendix

\section{Definition and basic properties of the $\delta$-constraints}\label{sec:defdelta}

In the text, we often use measures which are defined on $\R^n$, $n\ge 2$, by the formula
\begin{align}\label{eq:defgedeltameasure}
 \mu[\rmd s] = w(s) \,
 \delta\!\left(|s|^2 - N\right)\rmd^n s\,.
\end{align}
where $N>0$, $w:\R^n\to\R$ is a strictly positive continuous function, and 
$\rmd^n s$ denotes the Lebesgue measure on $\R^n$.  We first move to spherical coordinates to formally integrate out the $\delta$-function.  
Then for any continuous bounded non-negative function $f:\R^n \to \R$ we would have
\begin{align}\label{eq:defdelta}
& \int_{\R^n}\! \mu[\rmd s]\, f(s) 
 = \int_{\R^n}\!\rmd^n s\, f(s) w(s) \delta\!\left(|s|^2 - N\right)
 = \int_{S^{n-1}}\!\rmd \Omega\, \int_0^\infty \!\rmd r\, r^{n-1}
 f(s) w(s) \delta\!\left(r^2 - N\right) \nonumber \\
 &\quad  = \int_{S^{n-1}}\!\rmd \Omega\, \int_0^\infty \!\rmd t\, \frac{1}{2} t^{\frac{n}{2}-1}
 f(s) w(s) \delta\!\left(t - N\right) 
 = \frac{1}{2} N^{\frac{n}{2}-1}
 \int_{S^{n-1}}\!\rmd \Omega\, f(\sqrt{N}\Omega) w(\sqrt{N}\Omega) \,,
\end{align}
where
we have used shorthand notations $s=r \Omega=\sqrt{t}\Omega$ and the assumption that $N>0$.
Here $\rmd \Omega$ denotes the solid angle integration and thus its total mass is finite.  On the other hand, the values $\sqrt{N}\Omega$ cover the sphere with radius $\sqrt{N}$ and centre at the origin, which is a compact set.
Since the continuous function $f w$ is non-negative and has a maximum on this sphere, we may conclude that the 
map from $f$ to the right hand side of (\ref{eq:defdelta}) is a positive linear functional on 
the space of bounded continuous functions on $\R^n$.  Since $\R^n$ is a locally compact Hausdorff space, Riesz representation theorem implies that there is a unique regular Borel measure $\mu$ on $\R^n$ for which (\ref{eq:defdelta}) holds for all continuous $f$ with a compact support, and hence obviously also for all bounded continuous $f$.

This yields the definition of $\mu$ as a positive Radon measure.  The argument also shows that 
$\int_{\R^n}\! \mu[\rmd s]1 =  \frac{1}{2} N^{\frac{n}{2}-1}
 \int_{S^{n-1}}\!\rmd \Omega\,w(\sqrt{N}\Omega)<\infty$.  Since $w>0$ by assumption, and $S^{n-1}$ is compact, there is $c>0$ such that $w(\sqrt{N}\Omega)\ge c$.  Thus
the value of the integral is greater than zero, and it is always possible to normalize $\mu$ into a probability measure by multiplying $w$ with a positive constant, as was assumed in the text.

Consider the open set $E:=\defset{s\in\R^n}{|s|\ne \sqrt{N}}$, and define for all $j\in \N_+$
the closed sets $E_j := \defset{s\in\R^n}{||s|-\sqrt{N}|\ge \frac{1}{j}}$.  Clearly, $\cup_j E_j = E$, and by Urysohn's lemma 
to each $j$ there exists a continuous function $f_j$ such that $f_j(s)=1$ if $s\in E_j$, and $f_j(s)=0$ if $s\not\in E$.
We can use (\ref{eq:defdelta}) to compute $\int_{\R^n}\! \mu[\rmd s]\, f_j(s)$ and since $f_j(\sqrt{N}\Omega)=0$ for
all $\Omega$, it follows that 
\[
 0\le \mu(E) \le \sum_{j} \int_{\R^n}\! \mu[\rmd s]\, f_j(s)  =0\,.
\]
Therefore, $\mu(E)=0$ and $|s|^2=N$ almost surely under $\mu$, as claimed in the text.

Finally, let us point out that many ordinary properties
of Lebesgue measures are inherited by the measure $\mu$.  For instance, we are mainly interested in situations where 
$w$ and $f$ are continuous bounded functions on $\R^n$.  Then for any sequence $\vep_j>0$
for which $\vep_j\to 0$, we can approximate the value of $ \int\! \mu[ds] f(s)$ by replacing the $\delta$-function by a Gaussian function with a standard deviation $\vep_j$, i.e., if we define for $y\in \R$
\[
 G_\sigma(y) := (2\pi \sigma^2)^{-\frac{1}{2}} \rme^{-\frac{1}{2\sigma^2} y^2}\,,
\]
using spherical coordinates and dominated convergence theorem one may show that
\[
 \int_{\R^n}\! \mu[\rmd s]\, f(s) = \lim_{j\to\infty} \int_{\R^n}\!\rmd^n s\, f(s) w(s) G_{\vep_j}(|s|^2 - N)\,.
\]
Then, it is possible to perform a change of variables as usual to the Lebesgue integrals on the right hand side, and compute the limit to get the value of the left hand side.  Similarly, one may check that, if $w$ is 
invariant under permutation of the labels of the vector $s$ or rotations of the space $\R^n$,
then so is $\mu$.

In addition, the following two observations arising from the above limits are used in the text.
First,  
if one makes a scaling of the field $s$, the result follows standard formal rules of $\delta$-functions:
given $R>0$, make a change of variables $s=R s'$, yielding
\begin{align*}
 & \int_{\R^n}\!\rmd^n s\, f(s) w(s) G_\vep(|s|^2 - N) = 
 R^{n} \int_{\R^n}\!\rmd^n s'\, f(R s') w(R s') G_\vep(R^{2} |s'|^2 - N)\\
 & \quad =  R^{n} \int_{\R^n}\!\rmd^n s'\, f(R s') w(R s') R^{-2} G_{\vep R^{-2}}(|s'|^2 - N R^{-2})\,. 
\end{align*}
Therefore,
\[
 \int_{\R^n}\!\rmd^n s\, f(s) w(s) \,
 \delta\!\left(|s|^2 - N\right)
 = R^n \int_{\R^n}\!\rmd^n s'\, f(R s') w(R s') \,
 R^{-2} \delta\!\left(|s'|^2 - N R^{-2}\right)\,.
\]

Secondly, if $I\subset \set{1,2,\ldots,n}$, $2\le |I|<n$, we may use Fubini's theorem and spherical coordinates in $\R^I$
to integrate out the $\delta$-constraint.  Let $J$ denote the complement of $I$, set $m=|I|$, and apply Fubini's theorem to show that
\[
 \int_{\R^n}\!\rmd^n s\, f(s) w(s) G_\vep(|s|^2 - N)
 = \int_{\R^J}\!\rmd^{|J|} y
  \int_{S^{m-1}}\!\rmd \Omega\, \int_0^\infty \!\rmd r\, r^{m-1}
 f(s) w(s) G_\vep(r^2+ y^2 - N)\,.
\]
We change variables to $t=(r^2+ y^2 - N)/\vep$ and the right hand side becomes
\[
 \int_{\R^J}\!\rmd^{|J|} y
  \int_{S^{m-1}}\!\rmd \Omega\, \int_{(y^2 - N)/\vep}^\infty \!\rmd t\,  \frac{1}{2} (N-y^2+\vep t)^{\frac{m}{2}-1}
 f(s) w(s) \frac{1}{\sqrt{2\pi }} \rme^{-\frac{1}{2} t^2}\,.
\]
Since $|J|>0$, the set $y^2=N$ has zero Lebesgue measure and thus the integrand may be replaced by zero on this subset without changing the value of the integral.  The integral over the subset of $y$ with $y^2>N$, goes to zero as $\vep\to 0$, by the dominated convergence theorem.  Similarly, using dominated convergence theorem for values $y^2<N$ proves that
\[
 \int_{\R^n}\! \mu[\rmd s]\, f(s) = 
  \int_{\R^J}\!\rmd^{|J|} y \, \cf{|y|<\sqrt{N}} 
  \int_{S^{m-1}}\!\rmd \Omega\, \frac{1}{2} (N-y^2)^{\frac{m}{2}-1}
 f(s) w(s)\,,
\]
where $m=|I|$ and
$s=y+\sqrt{N-y^2}\,\Omega$, given in terms of the orthogonal decomposition
$\R^n = \R^J\oplus \R^I$.

\section{Coupling and Wasserstein distance}\label{sec:Wdistance}

We recall here the basic definitions and notions related to the main technical tool used in the proofs here, namely to couplings and the Wasserstein metric.
For readers interested in more detailed discussion and properties, we refer to the first few chapters of 
\cite{Villani:OptTransp}.

The Wasserstein metric is used to measure the distance between two probability measures on a Radon
space $X$.  The standard examples of Radon spaces are complete separable metric spaces, e.g., 
$\R^n$, separable Hilbert spaces, and their closed subsets.  We are only going to use Hilbert spaces here, i.e., assume that $X$ is a closed subset of a Hilbert space, and we consider the metric inherited from the norm $\norm{\cdot}$.

Suppose that $\mu_1$ and $\mu_2$ are Borel probability measures on $X$ such that there are $p\ge 1$ and $a_1,a_2\in X$ for which 
\[
 \int_X\! \mu_i(\rmd x) \norm{x-a_i}^p<\infty\,,\qquad i=1,2\,.
\]
A \defem{coupling} $\gamma$ between the measures $\mu_1$ and $\mu_2$ is a new probability measure on $X\times X$
such that its marginal distribution in the first variable is $\mu_1$ and in the second variable the marginal is $\mu_2$.  This occurs if and only if for all integrable Borel measurable functions $f:X\to \C$ we have $\mean{f(x_1)}_\gamma = \mean{f}_{\mu_1}$ and $\mean{f(x_2)}_\gamma = \mean{f}_{\mu_2}$ where $\gamma$-integration is taken over $(x_1,x_2)\in X\times X$, as in (\ref{eq:defWpdist}) below.
It is closely connected to coupling of two random variables in probability theory, although here there is less choice in 
the allowed $\sigma$-algebras.  Also, let us recall that if $X$ is a subset of a finite-dimensional space then it is locally compact, and thus by Riesz representation theorem it suffices to check that the above identities hold for all continuous and compactly supported functions $f$.

Under the above assumptions, the measures $\mu_1$ and $\mu_2$ have a finite $p$:th Wasserstein distance $W_p(\mu_1,\mu_2)$ 
which is defined via the formula
\begin{align}\label{eq:defWpdist}
 W_p(\mu_1,\mu_2)^p := \inf_{\gamma}\int_{X\times X}\! \gamma(\rmd x_1,\rmd x_2)\, \norm{x_1-x_2}^p 
\end{align}
where the infimum is taken over couplings $\gamma$ between $\mu_1$ and $\mu_2$.
There is always at least one such coupling, namely $\mu_1\times \mu_2$.  
Since $\norm{x_1-x_2}\le\norm{x_1-a_1}+\norm{a_1-a_2}+\norm{a_2-x_2}$, the expectation over $\gamma$ is finite for this coupling,
$\int_{X\times X}\! \gamma(\rmd x_1,\rmd x_2)\, \norm{x_1-x_2}^p<\infty$.

\end{document}